\documentclass[11pt]{article}

\usepackage{amsmath, amssymb, amsthm}
\usepackage{fullpage}
\usepackage{graphics}
\usepackage{environ}
\usepackage{mdframed}

\usepackage[labelfont=bf]{caption}

\newcommand\remove[1]{}

\theoremstyle{plain}
\newtheorem{lemma}{Lemma}[section]
\newtheorem{prop}{Proposition}[lemma]
\newtheorem{theorem}[lemma]{Theorem}

\newtheorem{thm}[lemma]{Theorem}
\newtheorem{conj}[lemma]{Conjecture}

\theoremstyle{definition}
\newtheorem{definition}[lemma]{Definition}

\DeclareMathOperator{\poly}{poly}
\newcommand\card[1]{\left| #1 \right|}
\newcommand\tup[1]{\langle #1 \rangle}
\newcommand\sett[2]{\left\{ \left. #1 \;\right\vert #2 \right\}}

\newcommand\set[1]{{\left\{ #1 \right\}}}

\newcommand\Prob[2]{{\Pr_{#1}\left[ {#2} \right]}}

\newcommand\Expc[2]{{\mathop{\bf E}_{#1}\left[ {#2} \right]}}
\newcommand\Var[2]{\mathop{\bf Var}_{#1}\left[{#2}\right]}

\newcommand\defeq{\doteq}

\newcommand\R{\mathbb{R}}
\newcommand\N{\mathbb{N}}

\newcommand\degree{d}
\newcommand\normal{\mathcal{G}}

\def\qed{\hfill $\vcenter{\hrule height .3mm
\hbox {\vrule width .3mm height 2.1mm \kern 2mm \vrule width .3mm
height 2.1mm} \hrule height .3mm}$ \bigskip}

\def \Sph{\mathbb{S}^{n-1}}
\def \RR {\mathbb R}

\def \EE {\mathbb E}

\def \Var {\mathrm{Var}}

\def \eps {\varepsilon}

\def \one {\mathrm{\mathbf{1}}}

\def \Id {\mathrm{I}_n}

\title{Reduction From Non-Unique Games To Boolean Unique Games}

\begin{document}

\author{
		  Ronen Eldan
		  \thanks{
		  {\tt ronen.eldan@weizmann.ac.il}.
		  Department of Mathematics,
          Weizmann Institute of Science. Supported by a European Research Council Starting Grant (ERC StG) and by an Israel Science Foundation grant no. 715/16.
		  } \and
        Dana Moshkovitz
        \thanks{
		  {\tt danama@cs.utexas.edu}.
        Department of Computer Science,
		  UT Austin. This material is based upon work supported by the National Science Foundation under grants number 1218547 and 1648712.}
}

\setcounter{page}{0}

\maketitle

\begin{abstract}
We reduce the problem of proving a ``Boolean Unique Games Conjecture'' (with gap $1-\delta$ vs.~$1-C\delta$, for any $C> 1$, and sufficiently small $\delta>0$) to the problem of proving a PCP Theorem for a certain non-unique game.
In a previous work, Khot and Moshkovitz suggested an inefficient {\em candidate} reduction (i.e., without a proof of soundness). The current work is the first to provide an efficient reduction along with a proof of soundness. The non-unique game we reduce from is similar to non-unique games for which PCP theorems are known.

Our proof relies on a new concentration theorem for functions in Gaussian space that are restricted to a random hyperplane. We bound the typical Euclidean distance between the low degree part of the restriction of the function to the hyperplane and the restriction to the hyperplane of the low degree part of the function.
\end{abstract}

\clearpage
\setcounter{page}{1}

\section{Introduction}

\subsection{The Unique Games Conjecture}

The Unique Games Conjecture was introduced by Khot~\cite{Khot} (see also the survey~\cite{KUGCSurvey}) in order to prove optimal inapproximability results that eluded existing techniques.
\begin{definition}[\textsc{Unique Game}]
The input of a unique game consists of a regular graph $G=(V,E)$, an alphabet $\Sigma$ of size $k$, and permutations $\pi_e:\Sigma\to\Sigma$ for the edges $e=(u,v)\in E$.
The task is to label each vertex with a symbol $\sigma(v)\in\Sigma$, as to maximize the fraction of edges $e=(u,v)\in E$ that are satisfied, i.e., $\pi_e(\sigma(u))=\sigma(v)$.
\end{definition}
The following two prover game describes a unique game instance: a verifier interacts with two all-powerful provers. The verifier picks uniformly an edge $e=(u,v)\in E$; sends $u$ to one prover and sends $v$ to the other prover. Each prover is supposed to respond with a label from $\Sigma$. The verifier accepts if the two received labels $\sigma(u),\sigma(v)$ satisfy $\pi_e(\sigma(u))=\sigma(v)$. Note that for every response of one prover in the game, there is a {\em unique} response of the other prover that is acceptable to the verifier. Hence, this two prover game is called a {\em unique game}.
The {\em value} of the game is the probability that the verifier accepts when the provers play optimally.


The Unique Games Conjecture says that it is NP-hard to distinguish unique games of value close\begin{footnote}{For unique games there is an efficient algorithm to distinguish games of value exactly $1$ from games of value smaller than $1$. Hence, it is necessary to focus on games of value close to $1$ rather than $1$.}\end{footnote} to $1$ from unique games of value close to $0$:
\begin{conj}[Unique Games Conjecture]
For every $\varepsilon,\delta>0$, there exists $k=k(\varepsilon,\delta)$, such that it is NP-hard, given a unique game instance with alphabet of size $k$, to distinguish between the case where at least $1-\delta$ fraction of the edges are satisfied and the case where at most $\varepsilon$ fraction of the edges are satisfied.
\end{conj}
We refer to the problem of distinguishing instances where at least $1-\delta$ fraction of the edges can be satisfied and instances where at most $\varepsilon$ fraction of the edges can be satisfied as $1-\delta$ vs. $\varepsilon$ unique games.

The Unique Games Conjecture is known to imply optimal NP-hardness of approximation for problems like \textsc{Max-Cut}~\cite{KKMO} and \textsc{Vertex-Cover}~\cite{KR2} that eluded optimal inapproximability results via existing techniques~\cite{Has97,Chan}. Moreover, under the Unique Games Conjecture one can prove inapproximability for wide families of approximation problems. Most notably, basic semidefinite programming (SDP)-based algorithms are optimal for all local constraint satisfaction problems~\cite{Rag}.

There are efficient algorithms for unique games in four cases: (i) Sufficiently small alphabet $k \leq \exp(1/\delta)$~\cite{Khot,CMM}; (ii) Sufficiently small $\delta = O(1/\log n)$ where $n$ is the size of the graph~\cite{Trevisan,GT,CMM,CMM2}; (iii) Large run-time $2^{n^{\poly(\delta)}}$~\cite{ABS}; (iv) Random-like structure of $G$~\cite{AKKSTV,KMM}.


There is an NP-hardness result for unique games for $\delta = 1/2$ and any $\varepsilon>0$ as follows from the recently proved 2-to-2 Theorem~\cite{KMS0,DKKMS0,DKKMS,BKS-ldt,KMMS,KMS}.
There is also a hardness result for any $\delta>0$ and $\varepsilon = 1 - 2\delta$~\cite{HHMOW,KMS} that holds in the Boolean case $k=2$.

The Boolean case $k=2$ is the first interesting case of unique games, and it captures problems like \textsc{Max-Cut} and \textsc{2Lin}(2).
The assignments to the variables are $\pm 1$, and each edge either requires its two endpoints to have the same assignment or different assignment.
It is conjectured (and, indeed, follows from the Unique Games Conjecture~\cite{KKMO}) that the best algorithm for Boolean unique games is the Goemans-Williamson SDP-based algorithm~\cite{GW} that can distinguish value $1-\delta$ from value $\varepsilon = 1-\Theta(\sqrt{\delta})$.
We focus on a weaker conjecture:
\begin{conj}[Boolean Unique Games Conjecture]\label{c:BUGC}
For every $C\geq 1$, for sufficiently small $\delta>0$, it is NP-hard to distinguish between unique games with $k=2$ where $1-\delta$ fraction of the edges can be satisfied, and ones where only $1-C\delta$ fraction of the edges can be satisfied.
\end{conj}
The Unique Games Conjecture can be thought of as an amplified version of Conjecture~\ref{c:BUGC}, with the soundness error close to $0$ rather than close to $1$ and the alphabet size appropriately increased. It is open whether the Unique Games Conjecture follows from Conjecture~\ref{c:BUGC}. There were past attempts to prove this implication via a ``strong parallel repetition'', but those attempts uncovered an obstacle~\cite{Raz-counter,BHHRRS}.

\subsection{This Work}

In a previous work Khot and Moshkovitz~\cite{KM-candidate} suggested a {\em candidate} reduction for proving hardness of $1-\delta$ vs. $1-C\delta$ Boolean unique games, however they could not prove the soundness of the reduction.
In this work we define a problem, Subspaces Near-Intersection, and show a provably sound reduction from Subspaces Near-Intersection to $1-\delta$ vs. $1-C\delta$ Boolean unique games.
Importantly, the NP-hardness of Subspaces Near-Intersection -- which we conjecture but do not prove -- is in the same spirit of known PCP Theorems, and resembles in many ways the 2-to-2 Theorem.
\begin{theorem}[Main Theorem]\label{t:main} Assume the Subspaces Near-Intersection Conjecture (Conjecture~\ref{c:sub-near-intersection} in the sequel).
For any $C\geq 1$, for any sufficiently small $\delta>0$, distinguishing $1-\delta$ vs.~$1-C\delta$ Boolean unique games is NP-hard. In fact,
if the Subspaces Near-Intersection problem requires time $T$, then distinguishing $1-\delta$ vs.~$1-C\delta$ Boolean unique games requires time $\Omega(T)$.
\end{theorem}
Our reduction has the added benefit of being highly efficient (linear-sized). In contrast, the reduction in~\cite{KM-candidate} had an exponential blowup, as it was only meant to rule out polynomial time algorithms for unique games under plausible assumptions on exponential hardness. Like for the 2-to-2 problem, one would expect a reduction from \textsc{Sat} to Subspaces Near-Intersection to map size-$n$ instances of \textsc{Sat} to size $n^{c(\delta)}$ instances of Subspaces Near-Intersection, where $\delta$ is the completeness error in Subspaces Near-Intersection and $c(\delta)\geq 1/\delta$ is a function of $\delta$.

Subspaces Near-Intersection is discussed in the next section.
The main ideas of the proof of Theorem~\ref{t:main} are discussed in Section~\ref{s:ideas}. A key lemma is a new concentration theorem for the restriction of a function in Gaussian space to a random hyperplane. The lemma bounds the Euclidean distance between the degree-$d$ part of the restriction and the restriction of the degree-$d$ part. The formal statement and more details appear in Section~\ref{s:concentration-intro}.

\subsection{Subspaces Near-Intersection Conjecture}

First we discuss existing PCP theorems (projection games), and a projection game based on \textsc{3Lin}($\R$), then we define the new conjecture.

\subsubsection{Projection Games}

Existing optimal hardness of approximation results follow from the {\em proven} NP-hardness of approximating {\em projection games}~\cite{AS,ALMSS,Raz,MR08}. In (the symmetric version of) projection games, the verifier tests the answer of each prover {\em separately} in a way that depends solely on the question to the prover, and then checks {\em equality} between parts of the two answers (the projections). For instance, given a \textsc{Sat} instance the verifier may ask each prover for the assignment to a subset of the variables. Each subset spans clauses and the verifier checks that those clauses are satisfied (a separate test for each prover that depends only on the question to the prover). The two subsets intersect, and the verifier checks that the provers agree on the assignments to the variables in the intersection (a comparison on parts of the answer).
Formally:
\begin{definition}[\textsc{Projection Game}]
The input of a projection game consists of a bi-regular graph $G=(X,Y,E)$ whose $X$-degree is denoted $q$, an alphabet $\Sigma$ and sets $L_x\subseteq \Sigma^q$ for every vertex $x\in X$.
The task is to label each vertex $x\in X$ with a symbol $\sigma(x)\in L_x$, as to maximize the probability that, when one picks $e=(x,y),(x',y)\in E$, it holds $\sigma(x)_y=\sigma(x')_y$.
Sometimes one describes the game over the graph $(X,\set{(x,x')})$.
\end{definition}

It is known that it is NP-hard to distinguish projection games of value $1$ from projection games of value close to $0$~\cite{AS,ALMSS,Raz,MR08}, and moreover that it requires time $2^{n^{1-o(1)}}$ assuming the widely believed Exponential Time Hypothesis\begin{footnote}{The Exponential Time Hypothesis postulates that \textsc{Sat} requires time $2^{\Omega(n)}$ on inputs of size $n$.}\end{footnote} as follows from an almost-linear sized reduction from \textsc{Sat} to projection games~\cite{MR08}.

2-to-2 games are projection games where given $\sigma(x)_y\in\Sigma$ there are only {\em two} possibilities for $\sigma(x)\in L_x\subseteq \Sigma^q$. It is known that it is NP-hard to distinguish 2-to-2 games of value close to $1$ from 2-to-2 games of value close to $0$~\cite{KMS0,DKKMS0,DKKMS,BKS-ldt,KMMS,KMS}. However, 2-to-2 games are easier than general projection games, since they have algorithms that run in time $2^{n^{\poly(\delta)}}$~\cite{ABS}. Appropriately, the known NP-hardness reduction to 2-to-2 games maps size $n$ inputs of \textsc{Sat} to size $n^{c(\delta)}$ 2-to-2 games for a function $c(\delta)\geq 1/\delta$.

\subsubsection{\textsc{3Lin}($\R$) Projection Game}\label{s:3LIN-projection}

Subspaces Near-Intersection is a proxy for the following projection game based on the Khot-Moshkovitz~\cite{KM} robust real \textsc{3Lin}: The verifier picks uniformly at random $100k$ real \textsc{3Lin} equations $c_1,\ldots,c_{100k}$ and two sets $S_1,S_2$ of $k$ variables among their variables, where $\card{S_1\cap S_2} = k-1$. Note that any subset of the linear equations induced on $S_1$ or on $S_2$ forms a linear subspace of $\R^k$. The verifier sends $S_1$ to one prover, and receives a unit vector that represents an assignment to $S_1$'s variables. The unit vector must satisfy a random linear constraint on $S_1$. The verifier sends $S_2$ to the other prover, and receives a unit vector that represents an assignment to $S_2$'s variables. The vector must satisfy a random linear constraint on $S_2$. The verifier projects each of the vectors on the $k-1$ coordinates that correspond to the intersection $S_1\cap S_2$, and measures the Euclidean distance between the projections.
Suppose that there exists a prover strategy where the projections are identical with probability $1-\delta$. The task is to efficiently compute a prover strategy that minimizes the average Euclidean distance between the projections.

Simple approximation algorithms for this problem guarantee distances $O(\sqrt{\delta/k})$ and $O(1/k)$:
\begin{itemize}
\item {\em Basic semidefinite programming} achieves {\em square distance} $\delta/k$, since in the completeness case one achieves deviation $0$ with probability $1-\delta$ and deviation $1/\sqrt{k}$ with probability $\delta$. As a result, this algorithm can efficiently guarantee distance $O(\sqrt{\delta/k})$.
\item {\em Correlated sampling} is the strategy in which the provers guess a clause in $S_1\cap S_2$, satisfy it (with a norm $1$ assignment) and assign all other coordinates $0$. It achieves distance $1$ with probability\begin{footnote}{Note that the error probability of correlated sampling can be made $C/k$ if one considers a projection onto a subspace of dimension $k-C$ instead of $k-1$.}\end{footnote} $1/k$, and deviation $0$ with the remaining probability.
\end{itemize}
Hence, the question is whether one can efficiently compute a prover strategy where the average distance between the projections is, say, $0.0001\cdot \min\set{\sqrt{\delta/k},1/k}$.

Subspaces Near-Intersection is closely related to this projection game: there one compares the vectors on their projection to a {\em generic} hyperplane in $\R^k$, as opposed to an {\em axis-parallel} hyperplane.

\subsubsection{Subspaces Near-Intersection}

The Subspaces Near-Intersection game is a projection game that is defined over the reals\begin{footnote}{The intention is to consider real numbers up to a finite precision, so the errors introduced by the finite precision are much smaller than any other quantity involved. For the sake of clarity in exposition we do not explicitly address precision errors.}\end{footnote}. Each vertex is associated with a linear subspace in $\R^k$, and a labeling to the vertex is a unit vector that satisfies the constraints. Each edge is associated with a hyperplane in $\R^k$. The vectors on the endpoints of the edge should have the same restriction to the hyperplane of the edge.

\begin{definition}[Subspaces Near-Intersection]\label{d:subspaces-intersection}
The input is a regular graph $G=(V,E)$, $k\times k$ matrices $A_v$ with entries in $[-1,1]$ for the vertices $v\in V$, and unit vectors $\Theta_e\in \R^k$ for the edges.
We assume that, per vertex $v\in V$, when one picks a uniform edge $e=(u,v)\in E$ that touches $v$, the vector $\Theta_e$ is uniform.
The task is to label each vertex with a unit vector $\sigma(v)\in\R^k$ such that $A_v\sigma(v)=0$, as to maximize the number of edges $e=(u,v)\in E$ with $Proj_{\Theta_e^\perp}(\sigma(u))=Proj_{\Theta_e^\perp}(\sigma(v))$ (``satisfied edges'').
We say that the edge is $\alpha$-satisfied if $| Proj_{\Theta_e^\perp}(\sigma(u))-Proj_{\Theta_e^\perp}(\sigma(v))|_2\leq \alpha$.
\end{definition}

As before, in the case that there exists an assignment where the distance between the projections is $0$ with probability $1-\delta$ and $1/\sqrt{k}$ with probability $\delta$, a semidefinite programming algorithm that minimizes the square distance between the projections, would lead to distance $\sqrt{\delta/k}$ between the projections.
There is a natural matching semidefinite programming integrality gap for Subspaces Near-Intersection described in Appendix~\ref{s:sdp-ig}.
The correlated sampling algorithm we described for the \textsc{3Lin}($\R$) projection game in Sub-section~\ref{s:3LIN-projection} no longer applies.

There is an analogy between the games considered in the recent proof of the 2-to-2 Theorem and the Subspaces Near-Intersection game: in both games for every edge the label of one endpoint does not uniquely determine the label of the other endpoint, but rather {\em nearly} determines it, leaving out one ``degree of freedom''. In the 2-to-2 games of~\cite{KMS0,DKKMS0,KMS}, labels are vectors over the binary finite field, and one degree of freedom means that there are two possibilities for the answer of the other prover. Here labels are real vectors and one of their ``coordinates'' remains undetermined.

\remove{
There is a semidefinite programming algorithm for Subspaces Near-Intersection that minimizes $\Expc{e=(u,v)\in E}{|Proj_{\Theta_e^\perp}(\sigma(u))-Proj_{\Theta_e^\perp}(\sigma(v))|_2^2}$. In the completeness case, this quantity is roughly $\delta$, and hence the algorithm can guarantee the same in the soundness case. Thus, the algorithm guarantees that for a typical edge $e=(u,v)\in E$ it holds that $|Proj_{\Theta_e^\perp}(\sigma(u))-Proj_{\Theta_e^\perp}(\sigma(v))|_2 \leq \Theta(\sqrt{\delta})$.
Also note that a typical coordinate in the unit vector $\sigma(v)\in\R^k$ is of magnitude $1/\sqrt{k}$. Therefore, in the completeness case, with probability $1-O(\delta)$ we have\begin{footnote}{As $k$ gets larger, the Subspaces Near-Intersection problem becomes closer to a unique game. Consequently, we suggest to focus on a moderate answer size, say $k=\tilde{\Theta}(1/\delta)$, for which the difference from a unique game is sufficiently large. Conveniently, in this regime of parameters known approximation algorithms for unique games fail~\cite{CMM}.}\end{footnote} $\card{\sigma(u) - \sigma(v)}_2^2 \leq O(1/k)$. Hence a semidefinite programming algorithm can even guarantee $\card{\sigma(u) - \sigma(v)}_2 \leq O(\sqrt{\delta + 1/k})$ for a typical edge $(u,v)\in E$ (i.e., proximity on the entire vector, not just the projection on the hyperplane).
The Subspaces Near-Intersection Conjecture is that it is NP-hard to obtain $|Proj_{\Theta_e^\perp}(\sigma(u))-Proj_{\Theta_e^\perp}(\sigma(v))|_2 \ll \Theta(\sqrt{\delta + 1/k})$.
}

For technical reasons, and similarly to the proof of the 2-to-2 Theorem, we will define a slight strengthening using {\em zoom-in}s.
For a linear subspace $Y\subseteq\R^k$ we define the $Y$-zoom-in Subspaces Near-Intersection game as follows: Focus on edges $e\in E$ where $Y\subseteq \Theta_e^{\perp}$, i.e., one can write $\Theta_e^\perp = Y + S_e$, where $S_e$ is a hyperplane in $Y^\perp$.
An edge is satisfied if $Proj_{S_e}(\sigma(u)) = Proj_{S_e}(\sigma(v))$ and is $\alpha$-satisfied if $\card{Proj_{S_e}(\sigma(u)) - Proj_{S_e}(\sigma(v))}_2\leq\alpha$.

\begin{conj}[Subspaces Near-Intersection Conjecture]\label{c:sub-near-intersection}
There exists a global constant $0<\alpha<1$, such that for any $\varepsilon,\delta>0$, $r\in\N$, there exists $k\geq 1$ such that $\sqrt{\delta/k}\gg 1/k$, and the following is NP-hard: The input is an instance of the Subspaces Near-Intersection problem. The task is to distinguish between the cases:
\begin{itemize}
\item {\em Completeness:} There exists a labeling $\sigma:V\to\R^k$ that satisfies\begin{footnote}{Near satisfaction suffices; see Section~\ref{s:future}.}\end{footnote} at least $1-\delta$ fraction of the edges $e=(u,v)\in E$. The remaining edges are $O(1/\sqrt{k})$-far from satisfied.

\item {\em Soundness:} For any $r$-dimensional $Y\subseteq \R^k$, for any labeling $\sigma:V\to\R^k$, the probability over the choice of $e=(u,v)$ in the $Y$-zoom-in, that $e$ is $\alpha\sqrt{\delta/k}$-satisfied is at most $\varepsilon$.
\end{itemize}
\end{conj}


\subsection{Main Ideas}\label{s:ideas}

This work builds on an idea suggested by Khot and Moshkovitz~\cite{KM-candidate} for proving hardness of unique games.
Like\begin{footnote}{The candidate reduction in~\cite{KM-candidate} had a variation on half-space encoding, namely, ${\sf interval}(\tup{a,x})$, where ${\sf interval}$ changes sign as one crosses any integer point, not just $0$. Crucially, we use half-spaces in the current paper.}\end{footnote}~\cite{KM-candidate} we replace the commonly used long code and Hadamard code by an encoding by half-spaces. We first explain the half-space idea, and then describe our new ideas in using and analyzing half-space encodings. 

The half-space defined by $a\in\R^k$ is $h_a:\R^k\to\set{\pm 1}$, where
$$h_a(x) = {\sf sign}(\tup{a,x}).$$
The half-space encoding of $a$ is the truth-table of $h_a$ where we enumerate over all $x\in\R^k$ up to a precision that makes the rounding error sufficiently smaller than any of the other quantities involved.

Half-space encoding is similar in structure to the Hadamard encoding, where a vector $a\in\set{0,1}^k$ is encoded as the linear function $l_a(x) = \tup{a,x}$ for all $x\in\set{0,1}^k$, and arithmetic is done over the finite field $\set{0,1}$. This similarity gains us two benefits that the Hadamard encoding has:
\begin{enumerate}
\item We can test linear conditions on $a\in\R^k$ by testing its encoding. Specifically, $\tup{a,c} =0$ for a vector $c\in\R^k$ iff $h_a(x+c) = h_a(x)$ for every $x\in\R^k$.
(On the soundness side we need $\card{\tup{a,c}}\gg 0$ to detect that the inequality does not hold; this the reason we require robustness).

\item Encodings of similar strings have common parts. Suppose that the projections of $a,a'\in\R^{k}$ on a hyperplane $\Theta^\perp$ are the same. Then, when one picks $x\in\Theta^\perp$ it holds that $\tup{a,x} = \tup{a',x}$. Importantly, the union of all hyperplanes covers $\R^k$ uniformly.

\end{enumerate}
Note that both equations $h_a(x+c) = h_a(x)$ and $h_a(x) = h_{a'}(x')$ are unique tests.
We remark that a property like the first is used in any optimal inapproximability result that uses the Hadamard code, and a property like the second was used in the proof of the 2-to-2 Games Theorem (under the name ``sub-code covering'').
Crucially, half-space encoding has a property that the Hadamard encoding does not have, but the long code does have, namely, a unique test:
\begin{enumerate}
\item[3.] Noise stability test. Half-spaces optimize the success probability of the following test: pick random Gaussian $x\in\R^k$, perturb $x$ to obtain $x'\in\R^k$ also distributed as a Gaussian. Check whether $h_a(x) = h_a(x')$.
\end{enumerate}

In discrete space, the long code encoding $d_i(x) = x_i$ optimizes the analogous noise stability test, and this was used to show hardness of Boolean unique games assuming the Unique Games Conjecture~\cite{KKMO}.

In~\cite{KM-candidate} it was suggested that to prove NP-hardness of Boolean unique games one needs robustness of the noise stability test:
\begin{quote}
Suppose that a half-space passes the noise stability test with probability $1-\delta$. Assume that a balanced function $f:\R^k\to\set{\pm 1}$ passes the test with probability $1-C\delta$ for $C>1$. Does $f$ correspond to a half-space?
\end{quote}
Works that dealt with robustness in noise stability~\cite{MN0,MN,E2} proved such results for functions that pass the test with probability at least $1-\delta-\epsilon$ for $\epsilon\ll\delta$. Such must be the same as a half-space almost everywhere.
When the acceptance probability is $1-C\delta$, the function $f$ can have many forms, including functions of $C$ half-spaces, low degree threshold functions, and many more. In particular, the function may have no correlation with any half-space.
Mossel and Neeman~\cite{MN3} note that functions that pass the noise stability test with constant probability have to correlate with a half-space {\em after a large random shift}, but we are unable to use this fact since a shift hurts the second property above.

Our idea is not to focus on a half-space that correlates with $f$ (which corresponds to the linear part of $f$), but rather consider the {\em low degree part} of $f$ (where the low degree part is obtained from the Hermite expansion of $f$). By the noise stability of $f$, its low degree part must be large. We argue about consistency between low degree parts of functions that are partly similar.
We also argue about the ability to extract vectors that satisfy linear tests from low degree parts that satisfy the same tests.

Crucially, all our estimates must be extremely tight, since the gap for Boolean unique games is extremely narrow to begin with, $1-\delta$ vs. $1-\Theta(\sqrt{\delta})$. We obtain the required tightness using two tools: hypercontractivity and concentration.

Hypercontractive inequalities (see, e.g.,~\cite{O}) bound norms of a ``smoothed'' function by norms of the original function. Here we use the Gaussian hypercontractive inequality, through the implied {\em level-$d$ inequalities} (see, e.g.,~\cite{O}), to show that Boolean functions that are the same with probability at least $1-\delta$ over the input must have low degree parts that are $\approx\delta$-close in $l_2$ distance. In contrast, a less careful estimate, not using Booleanity and hypercontractivity, only gives $\sqrt{\delta}$-closeness, which is useless in our context. Note that the functions we compare are restrictions of functions $f$ to hyperplanes (as in the second property above).

Concentration is discussed in Section~\ref{s:concentration-intro}. It considers functions restricted to a random hyperplane, and bounds the typical Euclidean distance of the low degree part of the restriction from the restriction of the low degree part. We use concentration to argue consistency between the low degree parts of the restrictions of a function to different hyperplanes. We note that the much easier to prove distance of $O(1/\sqrt{k})$ rather than $O(1/k)$ would have been useless for our application.

\subsection{Concentration of Degree-d Part}\label{s:concentration-intro}


Let $f:\R^n\to\R$, and let $f^{\leq \degree}$ be the degree-$\degree$ part of $f$. Note that $f^{\leq \degree}$ is a {\em global} property of $f$.
Let $\Theta$ be uniformly distributed in the $(n-1)$-dimensional sphere, so $\Theta^{\bot}$ is a random hyperplane in $\R^n$. Denote the restriction of $f$ to $\Theta^\perp$ by $f_{|\Theta^\perp}$. This is a {\em local} part of $f$.
We show a {\em local-to-global theorem}: the degree-$\degree$ part of $f_{|\Theta^\perp}$ is extremely close to the restriction of $f^{\leq \degree}$ to $\Theta^\perp$:

\begin{thm}[Concentration of degree-$d$ part]\label{t:concentration} For any $\varepsilon>0$, for every $0$-homogeneous\begin{footnote}{$f$ is $0$-homogeneous if $f(cx)=f(x)$ for every $x\in\R^n$ and $c>0$.}\end{footnote} function $f:\R^n\to\R$ with bounded $2$-norm, with probability at least $1-\varepsilon$ over $\Theta$,
$$\card{(f_{|\Theta^\perp})^{\leq\degree} - (f^{\leq \degree})_{|\Theta^\perp}}_2\leq {O}_{d,\varepsilon}({1}/{n}).$$
\end{thm}

Local-to-global theorems, like linearity testing~\cite{BLR} and low degree testing~\cite{RuSu} over finite fields, are key to PCP. 
With Theorem~\ref{t:concentration} we add a new, tight, low degree testing -type theorem, this time in the highly challenging case of real functions and approximate equality. To get intuition for why this case is so challenging, note that two different real low degree polynomials can be similar on much of the space (Carbery-Wright (Lemma~\ref{l:CW}) gives tight bounds). In contrast, two different low degree polynomials over a finite field are vastly different, and this is key to existing combinatorial and algebraic techniques, which we cannot use.
Standard analytic techniques (e.g., Hermite analysis, or a sampling theorem of Klartag and Regev~\cite{KlartagRegev}) give an upper bound of $O(1/\sqrt{n})$ rather than $O(1/n)$ even for $\degree=1$. As we remarked above, such bounds are useless for our needs.

Our proof is by a delicate second moment argument using symmetry considerations. Crucially, the second moment is a rotationally-invariant quadratic form in $f$, and hence we can use Schur's lemma from representation theory that classifies rotationally-invariant quadratic forms. The lemma implies that the second moment depends only on the spectrum of $f$, and not on its identity. Our calculations can therefore be significantly simplified by focusing on $f$ that depends only on one of its variables. Given a function that depends on one direction, the expression that we need to bound will only depend on the angle between this direction and $\Theta$. The technical bulk of the proof then amounts to showing that this dependence is quadratic in the scalar product, meaning that it is typically of the order $1/n$.

\subsection{The Road Ahead}\label{s:future}

This paper suggests two paths to NP-hardness of Boolean unique games:
\begin{enumerate}
\item Prove NP-hardness of Subspaces Near-Intersection as in Conjecture~\ref{c:sub-near-intersection}. This paper implies that NP-hardness of Boolean unique games would follow.
\item Lift the reduction in this paper to a reduction from the Khot-Moshkovitz NP-hard \textsc{3Lin}($\R$) to Boolean unique games. The reduction was outlined in Sub-section~\ref{s:3LIN-projection}.
\end{enumerate}
In this sub-section we give more details about each of these paths.

One can weaken the Subspaces Near-Intersection conjecture substantially and the analysis in this paper would still go through (with modifications): The verifier can project onto subspaces of dimension, say, $k-100$, instead of dimension $k-1$. In the completeness case there could be approximate equality (with deviation $O(\delta/\sqrt{k})$) rather than exact equality.
It is enough to have large soundness error, say $\varepsilon = 0.99$, instead of low error. The distance of the projections in the soundness case can be of the order of $\tilde{\Theta}(\delta/\sqrt{k} + 1/k)$, rather than $\Theta(\sqrt{\delta/k})$.

The reduction in this paper can be lifted to a reduction from a \textsc{3Lin}($\R$) projection game like we described in Sub-section~\ref{s:3LIN-projection} (instead of Subspaces-Near Intersection) to Boolean unique games. In this setting, we suggest to focus on projections onto subspaces of dimension sufficiently smaller than $k-1$, as to decrease the probability that the correlated sampling algorithm achieves distance $0$.
To analyze such a reduction one would need to address subspaces that are axes-parallel rather than generic, and this requires ideas beyond the ones in this paper. In particular, the concentration theorem we prove is no longer directly applicable.
In the authors' opinion, this path is the most promising path towards hardness of Boolean unique games.

\remove{
\subsection{Open Problems}

The main remaining open problems towards a proof of the Unique Games Conjecture are:
\begin{enumerate}
  \item Prove the Subspaces Near-Intersection Conjecture, possibly using ideas from the hardness of 2-to-2 games~\cite{KMS0,DKKMS0,DKKMS,BKS-ldt,KMMS,KMS} and Robust \textsc{3Lin}~\cite{KM}.
  \item Improve the hardness result for Boolean unique games in this paper from $1-\delta$ vs. $1-C\delta$ for any $C\geq 1$ to $1-\delta$ vs. $1-\Theta(\sqrt{\delta})$ matching the approximation algorithm of Goemans and Williamson~\cite{GW}.
  \item Amplify the hardness of unique games to $1-\delta$ vs. $\varepsilon$ for any $\delta,\varepsilon>0$ over a large alphabet $k=k(\varepsilon,\delta)$ as in the original Unique Games Conjecture~\cite{Khot}.
\end{enumerate}
}

\section{Preliminaries}\label{s:preliminaries}

\subsection{Hermite Polynomials}

Let $\normal^{n}$ denote the $n$-dimensional Gaussian distribution with
$n$ independent  mean-$0$ and variance-$1$ coordinates.
The space of all real functions $f:\R^n\to\R$
with $\Expc{x\sim\normal^{n}}{f(x)^2}<\infty$ is denoted $L^2(\R^n,\normal^{n})$. This is an inner product
space with inner product
 $$\left< f, g\right> \defeq  \Expc{x\sim\normal^{n}}{f(x)g(x)}.$$
For a natural number $j$, the $j$'th {\em Hermite polynomial}
$H_j:\R\to\R$ is
$$H_j(x) = \frac{1}{\sqrt{j!}}\cdot(-1)^j e^{x^2/2}\frac{d^j}{dx^j} e^{-x^2/2}.$$
The first few Hermite polynomials are $H_0\equiv 1$, $H_1(x) = x$, $H_2(x) = \frac{1}{\sqrt{2}} \cdot(x^2 - 1)$, $H_3(x) = \frac{1}{\sqrt{6}}\cdot(x^3 - 3x)$, $H_4(x) = \frac{1}{2\sqrt{6}}\cdot(x^4 - 6x^2 + 3)$.
The Hermite polynomials satisfy:
\begin{prop}[Orthonormality]\label{c:Hermite Orthogonality} For every
$j$, $\tup{H_j,H_j}=1$.  For every $i\neq j$, $\tup{H_i,H_j} = 0$.
In particular, for every $j\geq 1$, $\Expc{x\in \normal{}}{H_j(x)} = 0$.
\end{prop}
\remove{
\begin{prop}[Addition formula]\label{c:Hermite Addition}
$$H_j\left(\frac{x+y}{\sqrt{2}}\right) = \frac{1}{2^{j/2}}\cdot\sum_{k=0}^{j}
\sqrt{\binom{j}{k}}H_{k}(x)H_{j-k}(y).$$ 
\end{prop}
\begin{prop}[Recursive relation]
$$H_{i+1}(x) = xH_i(x) - i H_{i-1}(x).$$
\end{prop}
}
The multi-dimensional Hermite polynomials are:
$$H_{j_1,\ldots,j_n}(x_1,\ldots,x_n) = \prod_{i=1}^{n} H_{j_i}(x_i).$$
For multi-indices $L=(l_1,\ldots,l_n)$ and $T=(t_1,\ldots,t_n)$ we denote $L\leq T$ if $l_i\leq t_i$ for every $i$. We write $T-L$ to denote $(t_1-l_1,\ldots,t_n-l_n)$. We write $C^T$ to denote $C^{\sum_i t_i}$ and $\binom{T}{L}$ to denote $\binom{t_1}{l_1}\cdots\binom{t_n}{l_n}$.
\remove{
The multi-dimensional addition formula is:
\begin{prop}[Multi-dimensional addition formula]\label{c:Hermite Addition multi}
$$H_T\left(\frac{x+y}{\sqrt{2}}\right) = \frac{1}{\sqrt{2}^{T}}\cdot\sum_{L\leq T}
\sqrt{\binom{T}{L}}H_{L}(x)H_{T-L}(y).$$
\end{prop}
}
The Hermite polynomials form an orthonormal basis for the space $L^2(\R^n,\normal^{n})$. Hence, every
function $f\in L^2(\R^n,\normal^{n})$ can be written as
$$f(x) = \sum_{S\in\mathbb{N}^n} \hat{f}(S) \ H_S(x),$$
where $S$ is multi-index, i.e. an $n$-tuple of natural numbers, and
 $\hat{f}(S) \in \R$ (Hermite expansion). The
 size of a multi-index $S = (S_1,\ldots,S_n)$ is defined as
 $|S| = \sum_{i=1}^n S_i$. 
The degree-$d$ part of $f$ is $f^{= d} =
\sum_{\card{S}= d} \hat{f}(S) H_S(x)$.
The part of degree at most $\degree$ is $f^{\leq d} = \sum_{i=0}^{d} f^{=i}$.
When $f$ is anti-symmetric, i.e. $\forall x \in \R^n, f(-x)=-f(x)$,
we have $\widehat{f}(\vec{0})
 = \Expc{}{f} =  0$ ~and  $f^{\leq 0} \equiv 0$.

The noise operator (more commonly known as the
Ornstein-Uhlenbeck operator) $T_{\rho}$ takes a function $f\in
L^2(\R^n,\normal^{n})$ and produces a function $T_{\rho}f \in
L^2(\R^n,\normal^{n})$ that averages the value of $f$ over local
neighborhoods:
$$T_{\rho}f(x) = \Expc{y\in \normal^{n}}{f(\rho x + \sqrt{1-\rho^2}y)}.$$
The Hermite expansion of $T_{\rho}f$ can be obtained from the Hermite expansion of $f$ as follows:
\begin{prop}\label{c:Hermite Noise}
$$T_{\rho} f = \sum_{S} \rho^{\card{S}}\hat{f}(S) H_S.$$
\end{prop}

\subsection{Some classical inequalities}

The hypercontractive inequality is given in the next lemma.
\begin{lemma}[Hypercontractive inequality]\label{l:hypercontractivity} Let $f,g:\R^k\to\R$. 
For $0\leq \rho\leq \sqrt{rs}\leq 1$,
$$\tup{f,T_{\rho}g} \leq \card{f}_{1+r}\card{g}_{1+s}.$$
\end{lemma}
The inequality is often used to show the small sets cannot have much weight on low degree parts. Similarly, we will use a corollary of it to show that Boolean functions that are almost always the same must have low degree parts that are similar. The corollary is known as {\em level-$k$ inequality}:
\begin{lemma}[Level-$k$ inequality]\label{l:level-k} Let $f:\R^k\to\set{0,1}$ have mean $\Expc{}{f} = \alpha$ and let $k\leq 2\ln(1/\alpha)$. Then,
$$\card{f^{\leq k}}_2^2 \leq \left(\frac{2e}{k}\ln(1/\alpha)\right)^k\alpha^2.$$
\end{lemma}
A convenient re-formulation is
\begin{lemma}\label{l:level-k'} Let $A\subseteq\R^k$ be a set of probability $\alpha$. Let $p:\R^k\to\R$ be a polynomial of degree at most $k\leq 2\ln(1/\alpha)$ with $\card{p}_2 = 1$. Then, for $\chi_A$, the indicator function of $A$,
$$\card{\Expc{x}{p(x)\chi_A(x)}} \leq \left(\frac{2e}{k}\ln(1/\alpha)\right)^{k/2}\alpha.$$
\end{lemma}
\begin{proof} Since $p$ is of degree at most $k$, we have $\tup{\chi_A,p} = \tup{\chi_A^{\leq k},p}$.
By Cauchy-Schwarz inequality,
$$\tup{\chi_A^{\leq k},p}\leq \card{\chi_A^{\leq k}}_2\card{p}_2 \leq \card{\chi_A^{\leq k}}_2.$$
The lemma follows from a level-$k$ inequality (Lemma~\ref{l:level-k}) invoked on $\chi_A$.
\end{proof}

The Carbery-Wright anti-concentration inequality shows that a low degree polynomial cannot be concentrated around any point:
\begin{lemma}[Carbery-Wright Anti-concentration~\cite{CW}]\label{l:CW} For $t\in \R$ and $\varepsilon>0$, for a polynomial $p$ of degree $d$, $|p|_2=1$,
$$\Prob{x\sim\normal^n}{|p(x)-t|\leq\varepsilon}\leq O(d)\varepsilon^{1/d}.$$
\end{lemma}

The Gaussian Poincar\'{e} inequality upper bounds the variance of a function in terms of its derivative:
\begin{lemma}[Gaussian Poincar\'{e} inequality]\label{l:Poincare}
Let $f:\R^k\to\R$ have continuous derivatives. Then,
$$\Var{}{f} \leq \Expc{}{\card{\nabla f}^2}.$$
\end{lemma}

Klartag and Regev showed that a random subspace samples well any function: 
\begin{lemma}[Sampling~\cite{KlartagRegev}]\label{l:sampling}
Let $f:\R^k\to \R$ with $\card{f}_2 <\infty$. Let $0<\varepsilon<1$. Let $S$ be a uniform subspace of dimension $k-1$.
Then,
$$\Prob{S}{\card{\Expc{S}{f}-\Expc{}{f}}\geq \varepsilon\card{f}_2}\leq O\left(\exp\left(-\Omega\left(\frac{\varepsilon k}{\log(2/\varepsilon)}\right)\right)\right).$$
\end{lemma}
Their formulation referred to functions on spheres, but immediately implies the same for functions in Gaussian space by averaging over all possible radii. Their formulation referred to non-negative functions and multiplicative approximation, but immediately extends to general functions and additive approximation by separately considering the negative and positive parts of the function.

The next lemma follows from Lemma~\ref{l:sampling} (in fact, one needs a much weaker version of Lemma~\ref{l:sampling}):
\begin{lemma}\label{l:subspace-sampling} For any constants $0<\delta<1$ and $d\geq 1$,
For any subset $\mathcal{H}$ of fraction $\delta$ of $(k-1)$-dimensional subspaces in $\R^k$, the distribution induced on $d$-dimensional subspaces by picking $H\in\mathcal{H}$ and $S\subseteq H$, $dim(S)=d$, is $\tilde{O}_{d,\delta}(1/k)$-close in statistical distance to the uniform distribution over $d$-dimensional subspaces.
\end{lemma}

\section{Boolean Unique Game Construction}\label{s:ug}

Let $C\geq 1$. Fix an instance of the Subspaces Near-Intersection Problem, given by $G=(V,E)$, $k$, $\set{A_v}_v$, $\set{\Theta_e}_e$.
Let $\delta$ and $\varepsilon$ be the completeness and soundness errors, respectively, where $\delta>0$ is sufficiently small and $\varepsilon$ is a constant, say $1/10$. We will construct a Boolean unique games instance with completeness error $O(\delta/\sqrt{k})$ (where the $O(\cdot)$ hides a small absolute constant, independent of $C$) and soundness error $1-C\delta/\sqrt{k}$.

The unique game we construct consists of encodings of the labeling for the $v\in V$ via half-spaces.
\begin{definition}[half-space encoding]
The half-space encoding of $\sigma\in\R^k$ is the Boolean function $\R^k\to\set{\pm 1}$ defined as
$${\sf HS}_{\sigma}(x) = {\sf sign}(\tup{\sigma, x}).$$
\end{definition}

For every $v\in V$ and $x\in \R^k$ we have a unique game variable corresponding to $v,x$ that is supposed to be assigned ${\sf HS}_{\sigma(v)}(x)$ (The actual construction involves a discretization of $\R^k$ up to a very high precision in each coordinate. The precision depends on $k$ and $1/\delta$).
We denote by $f_{v}:\R^k\to\set{\pm 1}$ the actual assignment to the variables that correspond to $v$.

Next we group together variables in order to enforce certain basic structural properties on the $f_v$'s in a technique called {\em folding}. The properties we consider are ones that half-spaces have.

Half-spaces are anti-symmetric, i.e., for every $x\in\R^k$,
$${\sf HS}_{\sigma}(-x) = -{\sf HS}_{\sigma}(x).$$
While $f_v$ may not necessarily be ${\sf HS}_{\sigma(v)}$, we will enforce anti-symmetry by having only one variable for every pair of $x,-x$ where $x\in\R^k$.
\begin{definition}[anti-symmetry folding]\label{d:folding-odd}
In the unique games construction the functions $f_v$ satisfy
$f_v(-x) = -f_v(x)$ for every $x\in\R^k$.
\end{definition}
Half-spaces are $0$-homogeneous, i.e., for every $x\in\R^k$ and $c>0$ it holds ${\sf HS}_{\sigma}(c\cdot x) = {\sf HS}_{\sigma}(x)$. We enforce $0$-homogeneity as follows:
\begin{definition}[$0$-homogeneity folding]\label{d:folding-homogeneous}
In the unique games construction the functions $f_v$ satisfy
$f_v(cx) = f_v(x)$ for every $x\in\R^k$ and $c>0$.
\end{definition}
For every $A$ such that $A\sigma = 0$, for every $x,y\in\R^k$, $\alpha,\beta\in\R$, we have:
\begin{eqnarray*}
{\sf HS}_{\sigma}(\alpha xA + \beta y) & = & {\sf sign}(\tup{\sigma,\alpha xA + \beta y})\\
& = & {\sf sign}(\alpha\cdot\tup{\sigma,xA} + \tup{\sigma,\beta y})\\
& = & {\sf sign}(\alpha\cdot\tup{A\sigma,x} + \tup{\sigma,\beta y})\\
& = & {\sf sign}(\tup{\sigma,\beta y})
\end{eqnarray*}
Therefore we enforce:
\begin{definition}[constraints folding]\label{d:folding-constraints}
In the unique games construction the functions $f_v$ satisfy
$f_v(\alpha xA_v + \beta y) = f_v(\alpha zA_v+\beta y)$ for every $x,y,z\in\R^k$, $\alpha,\beta\in \R$.
\end{definition}

To complete the definition of the unique games instance, we define the equations over the variables.
The equations correspond to two local tests: (1) Noise test on $f_v$ for $v\in V$;
(2) Consistency test on $f_u,f_v$ for $(u,v)\in E$.
The equations are specified in Figure~\ref{f:verifier}.


\begin{figure}[h]
\begin{center}
\fbox{
\begin{minipage}{0.96\textwidth}
Verifier$\set{f_v}$\\
{\em Folding:} We assume that the $f_v$'s are folded as in Definitions~\ref{d:folding-odd},~\ref{d:folding-homogeneous} and~\ref{d:folding-constraints}.\\
Set $\beta = 1/(10^{10}C^2)$, $p=\delta/\sqrt{\beta k}$.    
The verifier performs the noise test with probability $p$; the consistency test with probability $1-p$:
\begin{itemize}
\item {\em Noise Test:} Pick at random $v\in V$.
    Pick $y,x,z\sim\normal^k$ and set $\tilde{x},\tilde{z} \in \R^k$ as follows: $\tilde{x} = (1-\beta)y + \sqrt{2\beta - \beta^2} x$,
    $\tilde{z} = (1-\beta)y + \sqrt{2\beta - \beta^2} z$. Check $f_{v}(\tilde{x}) = f_{v}(\tilde{z})$.

\item {\em Consistency Test:} Pick at random $e=(u,v)\in E$. Pick a random Gaussian $x\in \Theta_e^{\perp}$. Check $f_u(x) = f_v(x)$.
\end{itemize}\end{minipage}
} \caption{Unique game}\label{f:verifier}
\end{center}
\end{figure}

The size of the construction is linear in the size of the Subspaces Near-Intersection instance and a function of (the constants) $k$ and $1/\delta$.

\subsection{Completeness}

Suppose that there is an assignment $\sigma:V\to \R^k$ as in the completeness case of Subspaces Near-Intersection.
Further, assume that each $f_v$ corresponds to a half-space encoding of $\sigma(v)$.
The probability that the noise test rejects is $O(\sqrt{\beta})$ and it is performed with probability $p$, so its total contribution is $O(\delta/\sqrt{k})$. By the completeness of Subspaces Near-Intersection, with probability $1-\delta$ the consistency test always passes, and with probability $\delta$ it passes except with probability $1/\sqrt{k}$. Overall, the probability of rejection is $O(\delta/\sqrt{k})$.

\section{Soundness}\label{s:soundness}

Assume that $\set{f_v}_{v\in V}$ pass the unique tests with probability at least $1-C\delta/\sqrt{k}$. We will construct a constant-dimensional $Y\subset \R^k$ and an assignment $\sigma:V\to\R^k$. Each $\sigma(v)$ is a unit vector such that $A_v\sigma(v)=0$, and with constant probability over $e=(u,v)\in E_Y$, when one writes $\Theta_e^\perp = Y + S_e$ for $S_e$ orthogonal to $Y$, it holds that $$\card{ Proj_{S_e}(\sigma(u)) - Proj_{S_e}(\sigma(v)) }_2 \leq \tilde{O}_C(\delta/\sqrt{k} + 1/k),$$ where the $\tilde{O}_C(\cdot)$ hides logarithmic factors in $\sqrt{k}/\delta$, $k$, as well as factors that depend on $C$, and the deviation is therefore $\ll\sqrt{\delta/k}$.

The plan for the analysis is as follows:
Use the noise stability to decode a large low degree part for almost every vertex $v\in V$.
Use concentration to argue consistency between the restriction of the low degree part to an edge hyperplane and the low degree part of the restriction to the hyperplane, for most edges. The low degree parts of the restrictions to the edge hyperplane are close in $l_2$ distance for most edges thanks to the consistency test and hypercontractivity. Obtain from each low degree polynomial a vector by repeatedly differentiating the polynomial. The differentiation will be in random directions we pick, and we focus on zoom-in's so we can restrict to hyperplanes that contain the random directions.
We use consistency along edges to argue about consistency of the derivatives and of the number of differentiations.

For all $v\in V$ we have $|f_v|_2 = 1$.
By the success of the functions $f_v$ in the unique game, the noise test must pass except with probability $C\delta/(\sqrt{k}p)\leq C\sqrt{\beta}$ and the consistency test must pass except with probability $C\delta/(\sqrt{k}(1-p))\leq 2C\delta/\sqrt{k}$.
We say that $v\in V$ is {\em typical} if the noise test rejects with probability at most $100C\sqrt{\beta}$ when $v$ is chosen. In other words, for a typical $v\in V$,
$$\tup{f_v,T_{1-\beta}f_v} \geq 1- 200 C\sqrt{\beta}.$$
Note that all $v\in V$ are typical except for at most $0.1$ fraction.
We say that an edge $e=(u,v)\in E$ is {\em typical} if both $u$ and $v$ are typical and the consistency test rejects with probability at most $20C\delta/\sqrt{k}$ when $e$ is chosen. At least $0.7$ fraction of the edges are typical.

\subsection{Approximation By Low Degree}\label{s:low-deg-approx}

Our first lemma shows that the low degree part of a noise stable function approximates it:
\begin{lemma}[Noise stable functions have large low degree part]\label{l:large-low-deg}
Let $f:\R^k\to\R$, $\card{f}_2<\infty$. Let $0\leq \rho\leq 1$ and $\degree\geq 0$. Then,
$$|f^{\leq \degree}|_2^2 \geq \tup{f,T_{\rho}f} - \rho^{\degree} \card{f}_2^2.$$
\end{lemma}
\begin{proof}
We can decompose $f$ to its low degree part and its high degree part, $f = f^{\leq\degree } + f^{>\degree}$, and then $$\tup{f,T_{\rho}f} = \tup{f^{\leq \degree},T_{\rho}f^{\leq\degree}} + \tup{f^{> \degree},T_{\rho}f^{>\degree}}.$$
By Cauchy-Schwarz inequality,
$$\tup{f^{\leq \degree},T_{\rho}f^{\leq \degree}} \leq |f^{\leq \degree}|_2 |T_{\rho}f^{\leq\degree}|_2
\leq |f^{\leq \degree}|_2^2.$$
Therefore, by Proposition~\ref{c:Hermite Noise} and Parseval identity,
$$|f^{\leq \degree}|_2^2 \geq \tup{f^{\leq \degree},T_{\rho}f^{\leq \degree}} \geq \tup{f,T_{\rho}f} - \tup{f^{> \degree},T_{\rho}f^{>\degree}} > \tup{f,T_{\rho}f} - \rho^{\degree} \card{f}_2^2.$$
\end{proof}
Lemma~\ref{l:large-low-deg} implies that the low degree part of $f_v$ approximates $f_v$ for a typical $v\in V$:
$$|f_v^{\leq \degree}|_2^2 \geq 1-200 C\sqrt{\beta} - (1-\beta)^{\degree}.$$
In the above we used that $\card{f_v}_2 =1$.
We set $\degree = \Theta(1/\beta)$, so \begin{eqnarray}\label{e:large-low-deg}|f_v^{\leq \degree}|_2^2 \geq 0.99.\end{eqnarray}


\remove{
It will also be useful to get an upper bound on the low degree part of a function in terms of the function's noise stability:
\begin{lemma}[Noise sensitive functions have small low degree part]\label{l:small-low-deg}
Let $f:\R^k\to\R$. Let $0\leq \rho\leq 1$ and $\degree\geq 0$.
Then,
$$|f^{\leq \degree}|_2^2 \leq \rho^{-\degree}\tup{f,T_{\rho}f}.$$
\end{lemma}
\begin{proof}
We can write $f = f^{\leq\degree } + f^{>\degree}$.
$$\tup{f,T_{\rho}f} = \tup{f^{\leq \degree},T_{\rho}f^{\leq\degree}} + \tup{f^{> \degree},T_{\rho}f^{>\degree}}.$$
By the Hermite expansion of the noise operator, $\tup{f^{> \degree},T_{\rho}f^{>\degree}}\geq 0$, and, moreover,
$$\tup{f^{\leq\degree},T_\rho f^{\leq\degree}} \geq \rho^{\degree}|f^{\leq \degree}|_2^2.$$
The lemma follows.
\end{proof}
}

\subsection{Consistency of Degree-$\degree$ Parts}\label{s:consistency}

In this section we use the high acceptance probability of the consistency test in order to show that for most edges $(u,v)\in E$ the projections of the barycenters of $f_u$, $f_v$ onto $\Theta_e^\perp$ are extremely close to each other.
The proof uses the main technical tools we discussed in the introduction, namely hypercontractivity and concentration.

By hypercontractivity, Boolean functions that are the same except with probability $O(\delta)$ have low degree parts that are $\tilde{O}(\delta)$ apart in $2$-norm (note that there is a simple upper bound relying on Parseval identity alone, but it gives the worse upper bound $O(\sqrt{\delta})$), as proven in the following lemma:
\begin{lemma}[Low degree consistency]\label{l:low-deg-consistency}
Let $f,g:\R^k\to\set{\pm 1}$ be anti-symmetric functions. Let $0\leq \rho\leq 1$ and $\degree \leq 2\ln(1/\delta)$. Let $\delta>0$ be sufficiently small.
If $f(x) = g(x)$ with probability $1-\delta$ over Gaussian $x\in\R^k$, then
$|f^{\leq \degree} - g^{\leq \degree}|_2\leq 2\left(\frac{2e}{d}\ln(2/\delta)\right)^{d/2}\delta$.
\end{lemma}
\begin{proof}
We have $\card{f^{\leq \degree} - g^{\leq\degree}}_2 = \card{(f-g)^{\leq \degree}}_2$. Let $p$ be a polynomial of degree at most $d$ and $2$-norm $1$ that maximizes the correlation with $f-g$. Then, $\card{(f-g)^{\leq \degree}}_2 = \tup{f-g,p}$. Since $f$ and $g$ are anti-symmetric, so is $f-g$. Hence, $p$ is anti-symmetric.
Let $A\subseteq\R^k$ be the set of $x$ with $f(x)>g(x)$.  Since $f(x)>g(x)$ iff $g(-x)>f(-x)$, the probability of $A$ is $\delta/2$, and $\tup{f-g,p} = \Expc{x}{(2p(x)-2p(-x))\chi_A(x)} = 4\Expc{}{p(x)\chi_A(x)}$.
By Lemma~\ref{l:level-k'}, since $\degree \leq 2\ln(1/\delta)$ for sufficiently small $\delta>0$,
$$4\Expc{x}{p(x)\chi_A(x)} \leq 4\left(\frac{2e}{d}\ln(2/\delta)\right)^{d/2}(\delta/2).$$
The lemma follows by collecting all of the above.
\end{proof}
Let $(u,v)\in E$ be a typical edge.
By the consistency test, it holds that $f_{u{|\Theta_e^{\perp}}}(x) = f_{v|\Theta_e^{\perp}}(x)$ for random $x\in \Theta_e^{\perp}$ except with probability $O(\delta/\sqrt{k})$. Thus, by Lemma~\ref{l:low-deg-consistency},
\begin{eqnarray}
\left|(f_{u{|\Theta_e^{\perp}}})^{\leq\degree} - (f_{v|\Theta_e^{\perp}})^{\leq\degree}\right|_2 & \leq & \tilde{O}(\delta/\sqrt{k}).
\end{eqnarray}
By Theorem~\ref{t:concentration}, for each $v\in V$, for at least $0.99$ fraction of edges $e=(u,v)\in E$,
\begin{eqnarray}
\left|(f_{v|\Theta_e^{\perp}})^{\leq\degree} - (f_v^{\leq\degree})_{|\Theta_e^{\perp}}\right|_2 & \leq & O(1/k).
\end{eqnarray}
By the regularity of the graph, the triangle inequality and a union bound, with probability at least $0.6$ over $(u,v)\in E$, the edge is typical, and
\begin{eqnarray}\label{i:consistency}
  \nonumber\left|(f_u^{\leq\degree})_{|\Theta_e^{\perp}} - (f_v^{\leq\degree})_{|\Theta_e^{\perp}}\right|_2 & \leq &
  \left|(f_{u|\Theta_e^{\perp}})^{\leq\degree} - (f_u^{\leq\degree})_{|\Theta_e^{\perp}}\right|_2 +
  \left|(f_{v|\Theta_e^{\perp}})^{\leq\degree} - (f_v^{\leq\degree})_{|\Theta_e^{\perp}} \right|_2 \\
  & \leq & \tilde{O}(\delta/\sqrt{k} + 1/k).
\end{eqnarray}

\subsection{Defining The Assignment}\label{s:assignment}

In Section~\ref{s:consistency} we showed that for most edges $e=(u,v)\in E$ the degree-$d$ polynomials $f_u^{\leq\degree}$ and $f_v^{\leq\degree}$ are close over $\Theta_e^\perp$. In this section we show how to extract from the degree-$d$ polynomials {\em unit vectors} that satisfy the constraints and their projections onto $\Theta_e^\perp$ are close.

We next describe the main ideas behind the construction of unit vectors.
Close degree-$d$ polynomials, like $f_u^{\leq\degree}$ and $f_v^{\leq\degree}$ over $\Theta_e^\perp$, imply close degree-$1$ parts, and the degree-$1$ parts correspond to vectors in the linear subspaces associated with $u$ and $v$.
Hence, if the degree-$1$ parts of the polynomials were known to be of large $2$-norm, then one could have assigned each vertex its normalized linear part. Unfortunately, the degree-$1$ part of the polynomials can be $\vec{0}$.
The idea is to differentiate the degree-$d$ polynomials sufficiently many times until the degree-$1$ part is of sufficiently large $2$-norm.
The consistency deteriorates with the number of differentiations, but since the degree $d$ is constant, the number of differentiations is constant and the deterioration is limited.

To carry through the above plan we differentiate along random directions $y_1,\ldots,y_{d-1}$, and focus only on hyperplanes $\Theta_e^\perp$ that contain $Y=span\set{y_1,\ldots,y_{d-1}}$, since for those hyperplanes differentiation and restriction to $\Theta_e^\perp$ commute. This is the reason we focus on a zoom-in of the Subspaces Near-Intersection game.
This also introduces a certain asymmetry in favor of the directions in $Y$. To eliminate this asymmetry, we focus on random affine shifts of the space $Y^\perp$. The random choices of $Y$ and the shift would be useful in the analysis, but eventually we will fix them so they satisfy desired properties.

The assignment $\sigma:V\to\R^k$ for the Subspaces Near-Intersection instance is defined by the algorithm in Figure~\ref{f:assignment}. Our analysis closely follows the algorithm.




\begin{figure}[h]
\begin{center}
\fbox{
\begin{minipage}{0.96\textwidth}
{\em Global parameters:}
\begin{itemize}
  \item For sufficiently small constants $0<c_0<c_1<1$ (depending on the constant in Lemma~\ref{l:CW}), pick uniformly at random $$\eta\in \left[c_0 \cdot 2^{-2d\log d},c_1 \cdot 2^{-2d\log d}\right].$$
  \item Pick Gaussian vectors $y_1,\ldots,y_{\degree-1}\in\R^k$. Let $Y=span\set{y_1,\ldots,y_{\degree-1}}$.
  \item Pick Gaussian vector $y\in Y$.
\end{itemize}
For every typical $v\in V$ we define the assignment $\sigma(v)$ as follows (for other $v$'s leave $\sigma(v)$ undefined):
\begin{enumerate}
\item Let $D_v^{(0)} = f_v^{\leq\degree}$ and $i=0$.
\item Let $D_{v,y}^{(0)}:Y^\perp\to\R$ be the affine shift $D_{v,y}^{(0)}(x) = D_{v}^{(0)}(y+x)$

\item While $\card{(D_{v,y}^{(i)})^{=1}}_2^2 < \eta$,

\begin{enumerate}
  \item $i\leftarrow i+1$.

  \item Let $D_v^{(i)} = \frac{\partial}{\partial y_i}D_v^{(i-1)}$.
  \item Let $D_{v,y}^{(i)}:Y^\perp\to\R$ be the affine shift $D_{v,y}^{(i)}(x) = D_{v}^{(i)}(y+x)$.
\end{enumerate}
\item $i_v\leftarrow i$.
\item Let $vec_v\in Y^\perp$ be $(D_{v,y}^{(i_v)})^{=1}$.
\item $\sigma(v) \leftarrow \frac{vec_v}{\card{vec_v}_2}$.
\end{enumerate}
\end{minipage}
} \caption{The assignment $\sigma:V\to\R^k$ for the $Y$-zoom-in of Subspaces Near-Intersection}\label{f:assignment}
\end{center}
\end{figure}

The first lemma upper bounds the degree and lower bounds the norm on $D_{v}^{(i)}$ from the algorithm in Figure~\ref{f:assignment} for $0\leq i\leq d-1$:
\begin{lemma}[Norm lemma]\label{l:differentiation-lemma} For every typical $v\in V$, during the execution of the algorithm in Figure~\ref{f:assignment}, for every $0\leq i\leq \degree-1$,
\begin{enumerate}
  \item For all $y_1,\ldots,y_{i}$, the function $D_v^{(i)}$ is a polynomial of degree at most $d-i$.
  \item $\Expc{y_1,\ldots,y_{i}}{\card{\Expc{}{D_v^{(i)}}}^2}< \eta$.
  \item $\Expc{y_1,\ldots,y_{i}}{\card{D_v^{(i)}}_2^2} \geq 0.99 - \eta i$.
\end{enumerate}
\end{lemma}
\begin{proof}
We prove that the three items of the lemma hold by induction on $0\leq i\leq d-1$. First consider the case of $i=0$ where $D_v^{(0)} = f_v^{\leq\degree}$.

1. $f_v^{\leq\degree}$ is a polynomial of degree at most $d$.

2. By the anti-symmetry folding, $\Expc{}{f_v^{\leq\degree}}=0$.

3. By inequality~(\ref{e:large-low-deg}), for a typical $v$ we have $\card{f_v^{\leq d}}_2^2 \geq 0.99$.\\
Assume that the statement holds for $i-1$ and let us prove it for $i$.

1. The function $D_v^{(i)}$ is a polynomial of degree at most $deg(D_v^{(i-1)})-1$. The degree bound therefore follows from the inductive hypothesis.

2. $\Expc{}{D_v^{(i)}}$ is the constant part of $D_v^{(i)}= \tup{\nabla D_v^{(i-1)},y_i}$. Moreover, $\nabla D_v^{(i-1)}$ depends on $y_1,\ldots,y_{i-1}$ and is independent of $y_i$. Thus, $\Expc{}{D_v^{(i)}} = \tup{(D_v^{(i-1)})^{=1},y_i}$ is a normal variable with standard deviation $\card{(D_v^{(i-1)})^{=1}}_2$.
By the design of the algorithm, $\card{(D_v^{(i-1)})^{=1}}_2^2<\eta$ and hence $\Expc{y_1,\ldots,y_{d-1}}{\card{\Expc{}{D_v^{(i)}}}^2}< \eta$.


3. We have
$D_v^{(i)} = \tup{\nabla D_v^{(i-1)},y_i}$, where $\nabla D_v^{(i-1)}$ depends on $y_1,\ldots,y_{i-1}$ and is independent of $y_i$.
Thus, for every $x\in\R^k$, it holds that $D_v^{(i)}(x)$ is a normal variable with standard deviation $\card{\nabla D_v^{(i-1)}(x)}_2$.
Hence, $\Expc{y_1,\ldots,y_{d-1},x}{(D_v^{(i)})(x)^2} = \Expc{}{\card{\nabla D_v^{(i-1)}(x)}_2^2}$.
By the Gaussian Poincar\'{e} inequality (Lemma~\ref{l:Poincare}), for any $y_1,\ldots,y_{i}$,
$$\Expc{x}{\nabla D_v^{(i-1)}(x)^2} \geq \Var{D_v^{(i-1)}} = \card{D_v^{(i-1)}}_2^2 - \Expc{}{D_v^{(i-1)}}^2.$$
By the inductive hypothesis, $\Expc{}{\card{D_v^{(i-1)}}_2^2} \geq 0.99 - \eta (i-1)$ and $\Expc{}{D_v^{(i-1)}}^2 < \eta$. Hence,
$$\Expc{}{(D_v^{(i)}(x))^2} \geq 0.99 - \eta (i-1) - \eta = 0.99 - \eta i.$$
\end{proof}

By the following proposition and the constraints folding (see Definition~\ref{d:folding-constraints}), whenever $\sigma(v)$ is defined it satisfies $A_v\sigma(v)=\vec{0}$.
\begin{prop}\label{p:constraint-derive}
Let $f:\R^k\to \R$. If $f$ satisfies a constraints folding, then so do $f^{=i}$ for any $i$, any derivative of $f$, and any scalar multiplication of $f$.
\end{prop}

The next lemma uses Lemma~\ref{l:differentiation-lemma} to argue that $\sigma(v)$ is well-defined for most vertices $v\in V$.
\begin{lemma}[Assignment lemma]\label{l:assignment}
Let $v\in V$ be typical. With probability at least $0.99$ over $y_1,\ldots,y_{d-1}$ and $y$, the algorithm in Figure~\ref{f:assignment} terminates, $i_v$ is well-defined, and
$$\card{(D_{v,y}^{(i_v)})^{=1}}_2^2\geq \eta.$$
\end{lemma}
\begin{proof}
The algorithm terminates and $i_v$ is well-defined iff there exists $0\leq i\leq d-1$ such that $\card{(D_{v,y}^{(i)})^{=1}}_2^2\geq \eta$. Assume on way of contradiction that there is no such $i$.
By Lemma~\ref{l:differentiation-lemma}, when the algorithm reaches $i=d-1$, the polynomial $D_v^{(d-1)}$ is of degree $1$ and $$\Expc{Y,y}{\card{D_{v,y}^{(d-1)}}_2^2}=\Expc{y_1,\ldots,y_{d-1}}{\card{D_v^{(d-1)}}_2^2} \geq 0.9.$$
Since each coordinate of the coefficients vector $\nabla D_{v,y}^{(d-1)}$ is a polynomial of degree at most $d$ in $y_1,\ldots,y_{d-1}$ and $y$, the norm $\card{D_{v,y}^{(d-1)}}_2^2$ is a polynomial of degree at most $2d$ in $y_1,\ldots,y_{d-1},y$.
By convexity, $$\Expc{y_1,\ldots,y_{d-1},y}{\card{D_{v,y}^{(d-1)}}_2^4}\geq \left(\Expc{}{\card{D_{v,y}^{(d-1)}}_2^2}\right)^2 \geq 0.81.$$
By Carbery-Wright anti-concentration (Lemma~\ref{l:CW}), $\card{D_{v,y}^{(d-1)}}_2^2\geq \eta$ with probability at least $0.99$ over $y_1,\ldots,y_{d-1}$ and $y$. In this case, the loop in the algorithm in Figure~\ref{f:assignment} terminates and $i_v=d-1$.
\end{proof}


The next lemma argues consistency between $D_u^{(i)}$ and $D_v^{(i)}$ across most edges $e=(u,v)\in E$, provided that $y_1,\ldots,y_{d-1}\in\Theta_e^\perp$ (note that the degree $d$ is constant so the large dependence in $d$ -- which we state here explicitly, and later omit in the $O(\cdot)$ notation -- is permissible).
\begin{lemma}[Consistency lemma]\label{l:consistency} With probability at least $0.6$ over $e=(u,v)\in E$,
for every $0\leq i\leq d-1$,
\begin{eqnarray*}
\Expc{y_1,\ldots,y_{i}\in \Theta_e^\perp}{\left|(D_u^{(i)})_{|\Theta_e^{\perp}} -  (D_v^{(i)})_{|\Theta_e^{\perp}}\right|_2} & \leq & (O(d))^{i}\cdot\tilde{O}(\delta/\sqrt{k} + 1/k).
  \end{eqnarray*}
\end{lemma}
\begin{proof} By induction over $i$.
For $i=0$, the inequality follows from inequality~(\ref{i:consistency}): for at least $0.6$ of the edges $e=(u,v)\in E$ we have
$$\left|(f_u^{\leq\degree}-f_v^{\leq\degree})_{|\Theta_e^{\perp}}\right|_2 \leq \tilde{O}(\delta/\sqrt{k} + 1/k).$$
Assume that the claim holds for $i-1$, and let us prove it for $i$. Let $(u,v)\in E$.
We have
$D_u^{(i)} - D_v^{(i)} = \tup{\nabla (D_u^{(i-1)}- D_v^{(i-1)}),y_i}$, where $\nabla (D_u^{(i-1)}- D_v^{(i-1)})$ depends on $y_1,\ldots,y_{i-1}$ and is independent of $y_i$.
Thus, for every $y_1,\ldots,y_{i-1}$ and $x\in\R^k$, it holds that $(D_u^{(i)} - D_v^{(i)})(x)$ is a normal variable with standard deviation $\card{\nabla (D_u^{(i-1)}- D_v^{(i-1)})(x)}_2$.
Thus, by concavity and the inductive hypothesis,
\begin{eqnarray*}
  \Expc{y_1,\ldots,y_i\in \Theta_e^\perp}{\sqrt{\Expc{x\in \Theta_e^\perp}{(D_u^{(i)} - D_v^{(i)})(x)^2}}} & \leq & \Expc{y_1,\ldots,y_{i-1}}{\sqrt{\Expc{x,y_i}{\card{\nabla(D_u^{(i-1)}- D_v^{(i-1)})(x)}_2^2}}}\\
    & \leq & O(d)\cdot \Expc{y_1,\ldots,y_{i-1}}{\sqrt{\Expc{x}{(D_u^{(i-1)}- D_v^{(i-1)})(x)^2}}}\\
    & \leq & (O(d))^{i}\tilde{O}(\delta/\sqrt{k} + 1/k).
\end{eqnarray*}
\end{proof}

The next lemma is similar to Lemma~\ref{l:consistency}, but applies to the shifted $D_{u,y}^{(i)}$ and $D_{v,y}^{(i)}$ rather than to $D_{u}^{(i)}$ and $D_{v}^{(i)}$.
Recall that $Y = span\set{y_1,\ldots,y_{d-1}}$ and $E_Y = \sett{e\in E}{Y\subseteq \Theta_e^\perp}$. For each $e\in E_Y$ we write $\Theta_e^\perp = Y + S_e$. The subspace $S_e$ is a uniform hyperplane in $Y^\perp$.
\begin{lemma}\label{l:consistency'} With probability at least $0.99$ over $Y$ and $y$, with probability at least $0.6$ over $e=(u,v)\in E_Y$,
for every $0\leq i\leq d-1$,
\begin{eqnarray*}
\left|(D_{u,y}^{(i)})_{|S_e^{\perp}} -  (D_{v,y}^{(i)})_{|S_e^{\perp}}\right|_2 & \leq & \tilde{O}(\delta/\sqrt{k} + 1/k).
\end{eqnarray*}
\end{lemma}
\begin{proof}
By Lemma~\ref{l:consistency}, with probability at least $0.6$ over $e=(u,v)\in E$, for every $0\leq i\leq d-1$,
\begin{eqnarray}
\Expc{Y\subseteq \Theta_e^\perp}{\sqrt{\Expc{y\in Y,x\in S_e}{(D^{(i)}_{u,y} - D^{(i)}_{v,y})(x)^2}}} & \leq & \tilde{O}(\delta/\sqrt{k} + 1/k).
\end{eqnarray}
By concavity, with probability at least $0.6$ over $e=(u,v)\in E$, for every $0\leq i\leq d-1$,
\begin{eqnarray}
\Expc{Y\subseteq \Theta_e^\perp}{\Expc{y\in Y}{\sqrt{\Expc{x\in S_e}{(D^{(i)}_{u,y} - D^{(i)}_{v,y})(x)^2}}}} & \leq & \tilde{O}(\delta/\sqrt{k} + 1/k).
\end{eqnarray}
By Markov's inequality, with probability at least $0.6$ over $e=(u,v)\in E$, with probability at least $0.99$ over $Y\subseteq \Theta_e^\perp$ and $y \in Y$, we have
\begin{eqnarray}
\left|(D^{(i)}_{u,y})_{|S_e} - (D^{(i)}_{v,y})_{|S_e} \right|_2 & \leq & \tilde{O}(\delta/\sqrt{k} + 1/k).
\end{eqnarray}

By Lemma~\ref{l:subspace-sampling}, the distribution induced on $e$ and $Y$ by first picking $e\in E$ out of the set of fraction $0.6$, and then picking $Y\subseteq \Theta_e^\perp$, is close to the distribution that picks $Y$ by picking Gaussian $y_1,\ldots,y_{d-1}$, $Y=span\set{y_1,\ldots,y_{d-1}}$, and then picks $e\in E_Y$ that belongs to the set of fraction $0.6$.
Therefore, with probability $0.99$ over $Y,y$, the above event also holds with probability $0.6$ over $e \in E_Y$.
\end{proof}

By Lemmas~\ref{l:assignment} and~\ref{l:consistency'}, there exist $y_1,\ldots,y_{\degree-1}$ and $y$, such that with probability at least $0.5$ over $e=(u,v)\in E_Y$, the following two conditions holds
(recall that when one picks $e=(u,v)\in E_Y$ uniformly, the distribution over $v$ is uniform over $V$, and that $0.9$ fraction of the vertices $v\in V$ are typical):
\begin{enumerate}
\item $\card{(D_{v,y}^{(i_v)})^{=1}}_2^2\geq \eta$.

\item For every $0\leq i\leq d-1$,
$\left|(D_{u,y}^{(i)})_{|S_e^{\perp}} - (D_{v,y}^{(i)})_{|S_e^{\perp}}\right|_2 \leq \tilde{O}(\delta/\sqrt{k} + 1/k)$.
\end{enumerate}
The second item implies that for every $0\leq i\leq d-1$, $\left|((D^{(i)}_{u,y})_{|S_e})^{=1} - ((D^{(i)}_{v,y})_{|S_e})^{=1} \right|_2  \leq \tilde{O}(\delta/\sqrt{k} + 1/k)$.
The case $\degree=1$ of Theorem~\ref{t:concentration} implies that for every $u\in V$ with probability at least $0.999$ over the edge $e=(u,v)\in E_Y$, for every $i$,
\begin{eqnarray}\label{i:degree-1-parts-consistency}
\card{((D_{u,y}^{(i)})_{|S_e})^{=1} - (D_{u,y}^{(i)})^{=1})_{|S_e} }_2  & \leq & \tilde{O}(\delta/\sqrt{k} + 1/k).
\end{eqnarray}
Applying the same to $v\in V$ and taking a union bound and a triangle inequality, with probability at least $0.49$ over $(u,v)\in E_Y$, for every $i$,
\begin{eqnarray}\label{i:degree-1-parts-consistency}
\card{((D_{u,y}^{(i)})^{=1})_{|S_e} - ((D_{v,y}^{(i)})^{=1})_{|S_e} }_2  & \leq & \tilde{O}(\delta/\sqrt{k} + 1/k).
\end{eqnarray}
Note that inequality~(\ref{i:degree-1-parts-consistency}) implies consistency between vectors corresponding to $u$ and to $v$ restricted to the hyperplane of interest. It remains to argue that $i_u = i_v$ with high probability.
As a consequence of inequality~(\ref{i:degree-1-parts-consistency}), with probability at least $0.49$ over $e=(u,v)\in E_Y$, for every $i$,
\begin{eqnarray}\label{i:norm-consistency}
\card{\card{((D_{u,y}^{(i)})^{=1})_{|S_e}}_2^2 - \card{((D_{v,y}^{(i)})^{=1})_{|S_e}}_2^2} & \leq & \tilde{O}(\delta/\sqrt{k} + 1/k).
\end{eqnarray}
By sampling (Lemma~\ref{l:sampling}) and union bound, for every $u\in V$, with probability at least $0.999$ over $e=(u,v)\in E_Y$, for every $i$,
\begin{eqnarray}\label{i:sampling}
\card{\card{((D_{u,y}^{(i)})^{=1})_{|S_e}}_2^2 - \card{(D_{u,y}^{(i)})^{=1}}_2^2} & \leq & \tilde{O}(1/k).
\end{eqnarray}
A similar bound holds for $v$. Hence, from inequalities~(\ref{i:norm-consistency}) and (\ref{i:sampling}) via a union bound and a triangle inequality, with probability at least $0.47$ over $e=(u,v)\in E_Y$, for every $i$,
\begin{eqnarray}\label{i:degree-1-norm-consistency-restricted}
\card{\card{(D_{u,y}^{(i)})^{=1}}_2^2 - \card{(D_{v,y}^{(i)})^{=1}}_2^2} & \leq & \tilde{O}(\delta/\sqrt{k} + 1/k).
\end{eqnarray}
By the design of the algorithm in Figure~\ref{f:assignment}, inequality~(\ref{i:degree-1-norm-consistency-restricted}) guarantees that $i_u = i_v$ except with probability $\tilde{O}(\delta + 1/k)$. In this case, by inequality~(\ref{i:degree-1-parts-consistency}),
\begin{eqnarray}\label{l:proj-consistency}
\card{Proj_{S_e}(vec_{u}) - Proj_{S_e}(vec_{v})} & \leq & \tilde{O}(\delta/\sqrt{k} + 1/k).
\end{eqnarray}
The vectors $vec_u$ and $vec_v$ are normalized to obtain $\sigma(u)$ and $\sigma(v)$, respectively. Hence, by inequalities~(\ref{l:proj-consistency}) and (\ref{i:degree-1-norm-consistency-restricted}), and since $\card{(D_{u,y}^{(i_u)})^{=1}}_2\geq \Omega(1)$, with probability at least $0.47$ over $e=(u,v)\in E_Y$,
$$\card{Proj_{S_e}(\sigma(u)) - Proj_{S_e}(\sigma(v))} \leq \tilde{O}(\delta/\sqrt{k} + 1/k).$$

\section{Concentration of the restricted Hermite tensors}

In this section we prove Theorem~\ref{t:concentration}.

\textbf{Note:} In this section we use $n$ to denote the dimension. 

\subsection{Overview of the proof}
We first sketch some of the main steps of the proof. Consider the functional
$$
Q(f) = \EE_{\Theta} \card{(f_{|\Theta^\perp})^{\leq\degree} - (f^{\leq \degree})_{|\Theta^\perp}}_2^2
$$
where $\Theta$ is uniformly distributed in the sphere. It is not hard to check that $Q(f)$ is a quadratic form in $f$, which is invariant under compositions of $f$ with orthogonal transformations.

Here we allude to Schur's lemma, which states that rotational invariant quadratic forms on functions \emph{on the sphere} can be expressed as linear combinations of the $L_2$ norms of the projections onto eigenspaces of the Laplacian. This means that the maximum of the quadratic form among functions with a perscribed $L_2$ norm must be attained on a function which only depends on the first coordinate $x_1$.

Our quadratic form, however, is a functional of functions on $\RR^n$ rather than the sphere; this issue can be bypassed by considering homogeneous functions and using the concentration of the Gaussian in a thin spherical shell. Thus the first step of the proof roughly implies that it is sufficient to consider functions of the form $f(x_1,...,x_n) = g(x_1)$.

By applying rotations around the first vector of the standard basis, $e_1$, it is not hard to see that when $f$ is of the above form, the quantity
$$
\theta \to \card{(f_{|\theta^\perp})^{\leq\degree} - (f^{\leq \degree})_{|\theta^\perp}}_2^2
$$
only depends on $\theta_1 := \langle \theta, e_1 \rangle$. By concentration of measure, this angle is of the order $1/\sqrt{n}$. The technical bulk of the proof is to show that the above expression behaves like $\theta_1^4$ for small $\theta_1$.

\subsection{Preliminaries}

A tensor $T$ of degree $\ell$ is identified with the multilinear polynomial
$$
T[x^1,...,x^{\ell}] = \sum_{ i_1,...,i_\ell \in [n]^{\ell} } T_{i_1,...,i_{\ell}} x^1_{i_1} \cdot ... \cdot x^{\ell}_{i_\ell}.
$$
For any $x \in \RR^n$, denote by $H^{(k)}(x)$, the $k$-th Hermite tensor associated with $x$, defined by
$$
H^{(k)}(x) := (-1)^k \phi(x)^{-1} (\nabla^{k} \phi(x)),
$$
where $\phi(x) = \exp(-|x|^2 / 2)$. For example, we have
$$
H^{(1)}(x) = x, ~~ H^{(1)}(x)[y] = \langle x, y \rangle,
$$
$$
H^{(2)}(x) = x^{\otimes 2} - \Id, ~~ H^{(1)}(x)[y,z] = \langle x,y \rangle \langle x,z \rangle - \langle y,z \rangle.
$$
and
$$
H^{(3)}(x)[y,z,w] = \langle x, y \rangle \langle x, z \rangle \langle x, w \rangle - \langle x, y \rangle \langle z, w \rangle - \langle x, z \rangle \langle y, w \rangle - \langle x, w \rangle \langle y, z \rangle,
$$
(see \cite[p. 157]{McCullagh}). For two tensors $T, U$ of degree $\ell$, define the Hilbert-Schmidt inner product by
$$
\langle T, U \rangle_{HS} = \sum_{(i_1,...,i_\ell) \in [n]^\ell} T_{i_1,...,i_\ell} U_{i_1,...,i_\ell}
$$
and the corresponding norm
$$
\| T \|_{HS}^2 = \langle T, T \rangle_{HS}.
$$
We will allow ourselves to abbreviate the notation and write $\|T\|$ and $\langle T,U\rangle$ whenever this causes no confusion. For a function $f$, we define its $k$-barycenter by
$$
b_k(f) := \int H^{(k)} (x) f(x) d \gamma(x)
$$
and also denote
$$
\alpha_k(f)^2 := \left \| b_k(f) \right \|_{HS}^2.
$$
For a tensor $T$ of degree $\ell$ and an orthogonal projection $P$, define
$$
P T [x_1,...,x_\ell] := T[P x_1,...,P x_\ell].
$$
It is not hard to verify that for $f:\RR^n \to \RR$ and for $\theta \in \Sph$, one has
\begin{equation}\label{eq:marginal}
P_{\theta^\perp} \int H^{(k)}(x) f(x) d \gamma(x) = \int H^{(k)} (x) f_\theta (x) d \gamma(x)
\end{equation}
where
$$
f_\theta(x) = \int_{-\infty}^{\infty} f(P_{\theta^{\perp}} x + t \theta  ) d \gamma(t)
$$
is the marginal of $f$ on $\theta^\perp$.

For a unit vector $\theta \in \Sph$, let $\gamma_\theta$ be the Gaussian measure restricted to $\{\langle x, \theta \rangle = 0\}$, in other words,
$$
d \gamma_\theta(x) = \frac{1}{(2 \pi)^{(n-1)/2}} e^{-| x|^2 / 2} \mathbf{1}_{\langle x, \theta \rangle = 0} d\mathcal{H}_{n-1}(x).
$$
For a function $f:\RR^n \to \RR$ define, by slight abuse of notation,
$$
b_k(f; \theta) = \int H^{(k)}(x) f(x) d \gamma_\theta.
$$
By the orthogonality of Hermite polynomials, we have for all
$$
f^{=k}(x) = \frac{1}{k!} \langle H^{(k)}(x), b_k(f) \rangle, ~~~ \forall x \in \RR^n.
$$
Likewise, for all $\theta \in \Sph$
$$
(f_{|\theta^\perp})^{=k} = \frac{1}{k!} \langle H^{(k)}(x), b_k(f; \theta) \rangle, ~~~ \forall x \in \theta^\perp.
$$
Therefore, by Parseval's identity, we have
$$
\card{(f_{|\theta^\perp})^{= k} - (f^{= k} )_{|\Theta^\perp}}_2 = \frac{1}{k!} \| P_{\theta^\perp} (b_k(f; \theta) - b_k(f)) \|_{HS}.
$$
Thus, for a function $f:\RR^n \to \RR$, we define
$$
Q(f) = Q_k(f) := \EE_{\theta \sim \sigma} \| P_{\theta^\perp} (b_k(f; \theta) - b_k(f)) \|_{HS}^2.
$$
Theorem \ref{t:concentration} will follow immediately from the next result.
\begin{theorem} \label{thm:conc}
	Let $f: \RR^n \to \RR$ be $0$-homogeneous.
	$$
	\EE_{\theta \sim \sigma} \| b_k(f; \theta) -  b_k(f)) \|_{HS}^2 = O_k(1/n^2).
	$$
\end{theorem}
\begin{proof}[Proof of Theorem \ref{t:concentration}]
	Apply Theorem \ref{thm:conc} for and $k \leq d$ and use Chebyshev's inequality and a union bound.
\end{proof}

\subsection{A reduction to functions depending on one variable}

The proof of the above theorem relies on the following lemma, which essentially reduces the problem to the case that $f$ is a low-degree polynomial which only depends on one variable.
\begin{lemma} \label{l:onedim}
	For any $0$-homogeneous function $f$ with $\|f\|_{L_2(\gamma)}=1$, there is a polynomial $h:\RR \to \RR$ of degree at most $8k$ such that, defining $\tilde f(x) = h \left ( \frac{x_1}{|x| / \sqrt{n}} \right )$, we have $\|\tilde f\|_{L_2(\gamma)} = 1$ and
	$$
	\left| Q_k(f) - Q_k(\tilde f) \right| = O(1/n^2).
	$$
\end{lemma}

The main step towards the lemma is the following proposition:
\begin{prop} \label{p:quad} Assuming that $f$ is $0$-homogeneous,
	There exists a polynomial $q$ on $\RR$, of degree at most $8k$, such that
	$$
	Q_k(f) = \int_{\Sph} \int_{\Sph} f(x) f(y) q(\langle x,y \rangle) d \gamma(x) d \gamma(y) + O(1/n^2).
	$$
\end{prop}

Before we prove Proposition~\ref{p:quad}, we need two additional propositions, whose proofs are deferred to the end of this section.
\begin{prop}\label{p:density}
	There exist constants $C_n$,$C_n'$ such that $C_n,C_n' < C$ for some universal constant $C>0$, and such that the following holds. Let $x,y\in\RR^n$ and let $\theta$ be uniformly distributed in $\Sph$. Then, for every continuous $g:\Sph \to \RR$,
	$$
	\lim_{\varepsilon \to 0} \frac{1}{\varepsilon^2} \EE \Bigl [\mathbf{1} \{ |\langle x, \theta \rangle| \leq \eps, |\langle y, \theta \rangle| \leq \varepsilon \} g(\theta) \Bigr ] = C_n \frac{1}{|x||y| \sqrt{1-\left \langle \frac{x}{|x|}, \frac{y}{|y|} \right \rangle^2 } } \EE g(\theta_1)
	$$	
	where $\theta_1$ is uniform in $\Sph \cap x^\perp \cap y^\perp$. Furthermore,
	$$
	\lim_{\varepsilon \to 0} \frac{1}{\varepsilon} \EE \left [ \mathbf{1} \{ |\langle x, \theta \rangle| \leq \eps \} g(\theta) \right ] = \frac{C_n'}{|x|} g(\theta_2)
	$$	
	where $\theta_2$ is uniform in $\Sph \cap x^\perp$.
\end{prop}

\begin{prop} \label{p:tensorpoly}
	For every $k,n \in \mathbb{N}$ there exist polynomials $p_1,p_2,p_3, p_4$ in $3$ variables, of degree at most $3k$, with coefficients bounded by $O_k(n^k)$, such that the following holds. For each $x,y \in \RR^n$, let $\theta_1$ be uniform in $\Sph \cap x^\perp \cap y^\perp$ and let $\theta_2$ be uniform in $\Sph \cap x^\perp$. Then we have the representations
	$$
	\EE \langle P_{\theta_1^\perp} H^{(k)}(x), P_{\theta_1^\perp} H^{(k)}(y) \rangle = p_1(|x|, |y|, \rho(x,y)) + \sqrt{1-\rho(x,y)^2} \cdot p_2 (|x|, |y|, \rho(x,y))
	$$
	and
	$$
	\EE \langle P_{\theta_2^\perp} H^{(k)}(x), P_{\theta_2^\perp} H^{(k)}(y) \rangle = p_3(|x|, |y|, \rho(x,y)) + \sqrt{1-\rho(x,y)^2} \cdot p_4 (|x|, |y|, \rho(x,y))
	$$
	where $\rho(x,y) := \left \langle \frac{x}{|x|},\frac{y}{|y|} \right  \rangle$.
\end{prop}

\begin{proof} [Proof of Proposition~\ref{p:quad}]
	By an approximation argument, we may assume that $f$ is continuous. We then have,
	$$
	\beta(f; \theta) = \lim_{\eps \to 0} \frac{\sqrt{2 \pi} }{2 \eps} \int \one \{ |\langle x, \theta \rangle| \leq \eps \} H^{(k)}(x) f(x) d \gamma.
	$$
	Therefore, we have
	\begin{align*}
	Q(f) & = \EE_{\theta \sim \sigma}  \left  \| \lim_{\eps \to 0} P_{\theta^\perp} \frac{\sqrt{2 \pi} }{2\eps} \int \one \{ |\langle x, \theta \rangle| \leq \eps \} H^{(k)}(x) f(x) d \gamma(x)  -  P_{\theta^\perp} \int H^{(k)}(x) f(x) d \gamma(x) \right \|_{HS}^2 \\
	& = \lim_{\eps \to 0} \left ( \EE_{\theta \sim \sigma} \left [ \frac{\pi}{2\eps^2} \int \one \left \{ { |\langle x, \theta \rangle| \leq \eps \atop |\langle y, \theta \rangle| \leq \eps } \right \} \langle P_{\theta^\perp}H^{(k)}(x), P_{\theta^\perp} H^{(k)}(y) \rangle  f(x) f(y) d \gamma(x,y) \right . \right . \\
	& ~~~~ - \frac{ \sqrt{2 \pi} }{\eps} \int \one \{ |\langle x, \theta \rangle| \leq \eps \} \langle P_{\theta^\perp}H^{(k)}(x), P_{\theta^\perp} H^{(k)}(y) \rangle  f(x) f(y) d \gamma(x,y) \\
	&~~~~ \left  . + \left  .  \int  \langle P_{\theta^\perp}H^{(k)}(x), P_{\theta^\perp} H^{(k)}(y) \rangle f(x) f(y) d \gamma(x,y)  \right ] \right ) \\
	& = \int  \left (h_1(x,y) - 2 h_2(x,y) + h_3(x,y) \right )  f(x) f(y) d \gamma(x,y),
	\end{align*}
	where
	$$
	h_1(x,y) = \lim_{\eps \to 0} \EE_{\theta \sim \sigma} \frac{\pi}{2\eps^2} \one \left \{ { |\langle x, \theta \rangle| \leq \eps \atop |\langle y, \theta \rangle| \leq \eps } \right \} \langle P_{\theta^\perp}H^{(k)}(x), P_{\theta^\perp} H^{(k)}(y) \rangle,
	$$
	$$
	h_2(x,y) = \lim_{\eps \to 0} \EE_{\theta \sim \sigma} \frac{\sqrt{2 \pi}}{\eps} \one \left \{ |\langle x, \theta \rangle| \leq \eps \right \} \langle P_{\theta^\perp}H^{(k)}(x), P_{\theta^\perp} H^{(k)}(y) \rangle
	$$
	and
	$$
	h_3(x,y) = \EE_{\theta \sim \sigma} \langle P_{\theta^\perp}H^{(k)}(x), P_{\theta^\perp} H^{(k)}(y) \rangle.
	$$
	By Proposition~\ref{p:density}, we have
	$$
	h_1(x,y) = \frac{C_n}{|x||y| \sqrt{1-\left \langle \frac{x}{|x|}, \frac{y}{|y|} \right \rangle^2 } } \EE_{\theta_1 \sim U(\Sph \cap x^\perp \cap y^\perp)} \langle P_{\theta_1^\perp}H^{(k)}(x), P_{\theta_1^\perp} H^{(k)}(y) \rangle
	$$
	for some constant $C_n$ depending only on the dimension, and which is smaller than a universal constant $C>0$. From this point on, the expression $C_k$ will denote a constant that depends only on $k$, whose value may vary between different instances.
	
	By Proposition~\ref{p:tensorpoly} there are polynomials $p_1,p_2$ of degree at most $3k$, with coefficients bounded by $C_k n^k$, such that
	$$
	h_1(x,y) = \frac{1}{|x||y|} \left (\frac{p_1(\rho(x,y), |x|,|y|)}{\sqrt{1-\rho(x,y)^2 }} + p_2(\rho(x,y),|x|,|y|) \right )
	$$
	where $\rho(x,y) = \left \langle \frac{x}{|x|},\frac{y}{|y|} \right  \rangle$.
	
	Since the coefficients of $p_1$ are bounded by $C_k n^k$, we have $p_1(\rho(x,y),|x|,|y|) \leq C_k n^k (|x|+1)^k (|y|+1)^k$. By taking the Taylor expansion of the function $s \to \frac{1}{\sqrt{1-s^2}}$ of order $2k+4$, we conclude that there exists a polynomial $q(\cdot)$ of degree $4k+4$ such that
	$$
	h_1(x,y) = \frac{q \left  ( \rho(x,y) ) p_1(\rho(x,y),|x|,|y|) + p_2(\rho(x,y),|x|,|y|) \right ) }{|x||y|}  + O_k \left ( n^k (1+|x|)^k (1+|y|)^k \rho(x,y)^{4k+4}  \right ).
	$$
	By Cauchy-Schwartz and since $\EE_{x \sim \gamma} |x|^{2k} \leq C_k n^k$ and $\EE_{x,y \sim \gamma} \left [  |\rho(x,y)|^{\ell}  \right ] \leq C_\ell n^{-\ell/2}$, we have
	$$
	n^k \int (|x|+1)^k (|y|+1)^k \rho(x,y)^{4k+4} f(x) f(y) d \gamma(x,y) \leq \frac{1}{n^2} C_k \|f\|_2^2.
	$$
	A combination of the last two displays imply that there exists a polynomial $q_1$ of degree at most $8k$ such that
	$$
	\int h_1(x,y) f(x) f(y) d \gamma(x,y) = \int q_1(\rho(x,y),|x|,|y|) d \gamma(x,y) + O_k \left (\frac{\|f\|_2^2}{n^2} \right ).
	$$
	Following a similar argument with the terms $h_2$ and $h_3$, we conclude that there exists a polynomial $p$ of degree at most $8k$ such that
	$$
	Q(f) = \int p \left (\rho(x,y), |x|,|y| \right )f(x)f(y)  d \gamma(x,y) + O_k \left (\frac{\|f\|_2^2}{n^2} \right ).
	$$
	Since $f$ is $0$-homogeneous, by polar integration one learns that for all $k_1,k_2,k_3$, there exist constants $C_{k_1,k_2,k_3}, C_{k_1,k_2,k_3}'$ such that
	\begin{align*}
	\int \int \langle x, y \rangle^{k_1} |x|^{k_2} |y|^{k_3} f(x) f(y) d \gamma(x) d \gamma(y) & = C_{k_1,k_2,k_3} \int \int \left \langle \frac{x}{|x|}, \frac{y}{|y|} \right \rangle^{k_1} f(x) f(y) d \gamma(x) d \gamma(y) \\
	& = C_{k_1,k_2,k_3}' \int_{\Sph} \int_{\Sph} \left \langle \frac{x}{|x|}, \frac{y}{|y|} \right \rangle^{k_1} f(x) f(y) d \sigma(x) d \sigma(y).
	\end{align*}
	We conclude that there exists a polynomial $q(\cdot)$ of degree at most $8k$ such that
	$$
	Q_k(f) = \int_{\Sph \times \Sph}  q \left ( \left \langle \frac{x}{|x|}, \frac{y}{|y|} \right \rangle \right ) f(x) f(y) d \sigma (x) d \sigma(y) + O_k(1/n^2).
	$$
\end{proof}

\begin{proof}[Proof of Lemma~\ref{l:onedim}]
	For a function $h \in L_2(\Sph)$, define by $\mathrm{Proj}_{\mathcal{S}_k} h$ the orthogonal projection of $h$ into the subspace spanned by spherical harmonics of degree $k$. An application of Schur's lemma (or the Funk-Hecke formula) ensures that for every polynomial $g$ degree $\ell$ there exist constant $\alpha_1,...,\alpha_\ell$ such that
	$$
	\int_{\Sph} \int_{\Sph} f(x) f(y) g(\langle x,y \rangle) d \gamma(x) d \gamma(y) = \sum_{i \leq \ell} \alpha_i \|\mathrm{Proj}_{\mathcal{S}_i} f\|_{L_2(\Sph)}^2
	$$
	Thus, by Proposition~\ref{p:quad} we learn that there are some $(\alpha_i)_{i=0}^{8k}$ such that
	\begin{equation}\label{eq:approxpoly}
	Q(f) = \sum_{0 \leq i \leq 8k} \alpha_i \|\mathrm{Proj}_{\mathcal{S}_i} f\|_{L_2(\Sph)}^2 + O_k(1/n^2).
	\end{equation}
	(in the last formula, by slight abuse of notation, on the right hand side the function $f$ should be understood as its restriction to the sphere). Now, for any $j \in \mathbb{N}$ there exists a function $h_j$ depending only on $x_1$ such that $\|\mathrm{Proj}_{\mathcal{S}_i} h_j\|_{L_2(\Sph)}^2 = \mathbf{1}_{\{i=j\}}$. Therefore, defining
	$$
	\tilde f(x) = \sum_{j} h_j\left ( \frac{x_1}{|x|} \right ) \|\mathrm{Proj}_{\mathcal{S}_i} f\|_{L_2(\Sph)},
	$$
	we have $\|\mathrm{Proj}_{\mathcal{S}_i} f\|_{L_2(\Sph)} = \|\mathrm{Proj}_{\mathcal{S}_i} \tilde f\|_{L_2(\Sph)}$ for all $i$, and therefore by \eqref{eq:approxpoly}, we have $|Q(f) - Q(\tilde f)| = O(1/n^2)$. Moreover, $\|f\|_{L_2(\gamma)} = \|f\|_{L_2(\Sph)} = \| \tilde f \|_{L_2(\Sph)}$. This completes the proof.
\end{proof}

\subsection{Finishing the proof}

Let $f: \RR^n \to \RR$ be a function which has the form
$$
f(x_1,...,x_n) = h \left (x_1 \frac{\sqrt{n}}{ |x|} \right ).
$$
for some polynomial $h:\RR \to \RR$ of degree at most $8k$ and with $\|f\|_{L_2(\gamma)} = 1$. In light of Lemma~\ref{l:onedim}, Theorem~\ref{thm:conc} will be concluded by showing that
\begin{equation}\label{eq:nts}
Q(f) = O_k(1/n^2).
\end{equation}
Let $\theta$ be uniform in $\Sph$. We first show that, by symmetry, we can essentially assume in our calculations that $\theta \in \mathrm{span}\{e_1,e_2\}$. Let us write $\theta_1 = \langle \theta, e_1 \rangle$ and define $\tilde{\theta} := e_1 \theta_1 + e_2 \sqrt{1-\theta_1^2}$. By symmetry of the function $f$ to orthogonal transformations which keep $e_1$ fixed, we have
$$
Q(f) = \EE_{\theta_1} \| P_{\tilde{\theta}^\perp} (b_k(f; \tilde{\theta}) - b_k(f)) \|_{HS}^2.
$$
In order to understand the role of the projection onto the subspace $\tilde{\theta}^\perp$, define an orthonormal basis to $\tilde{\theta}^\perp$ as follows: Set $e_1' = \sqrt{1-\theta_1^2} e_1 - \theta_1 e_2$ and $e_i' = e_{i+1}$ for $i=2,...,n-1$, so that $(e_i')_{i=1}^{n-1}$ form an orthonormal basis for $\tilde \theta^\perp$. We have,
\begin{equation}\label{eq:basisnorm}
\| P_{\tilde{\theta}^\perp} (b_k(f; \tilde{\theta}) - b_k(f)) \|_{HS}^2 = \sum_{(i_1,...,i_k) \in [n-1]^k} \left ( b_k(f; \tilde{\theta})[e_{i_1}',...,e_{i_k}'] - b_k(f)[e_{i_1}',...,e_{i_\ell}'] \right )^2.
\end{equation}
Fix $I = (i_1,...,i_\ell) \in [n-1]^\ell$. There exists a function $J_I$ and $\alpha(I) \in [k]$ such that
\begin{equation}\label{eq:pres1}
H^{(k)}(x)[e_{i_1}',...,e_{i_\ell}'] = H_{\alpha(I)}(\langle x, e_1' \rangle) J_{I}(\mathrm{Proj}_L(x)),
\end{equation}
where $L=\mathrm{span}(e_2',...,e'_{n-1})$. Let $\Gamma_1 \sim \mathcal{N}(0,1), \Gamma_2 \sim \mathcal{N}(0,1), \Gamma_3 \sim N(0, \mathrm{Proj}_L)$ be independent. In this case, note that
$$
e_1' \Gamma_1 + \tilde{\theta} \Gamma_2 + \Gamma_3 \stackrel{(d)}{=} \mathcal{N}(0, \mathrm{I}_n).
$$
We therefore have by equation \eqref{eq:pres1} and by the definition of $b_k(f; \tilde{\theta})$,
\begin{equation}\label{eq:term1decomposed}
b_k(f; \tilde{\theta})[e_{i_1}',...,e_{i_\ell}'] = \EE \left [  H_{\alpha(I)} (\Gamma_1) J_I(\Gamma_3) h \left ( \frac{ \sqrt{1-\theta_1^2} \Gamma_1  }{ \sqrt{(\Gamma_1^2 + |\Gamma_3|^2)/n }} \right ) \right ],
\end{equation}
and on the other hand,
\begin{equation}\label{eq:term2decomposed}
b_k(f)[e_{i_1}',...,e_{i_\ell}'] = \EE \left [  H_{\alpha(I)} (\Gamma_1) J_I(\Gamma_3) h \left ( \frac{ \sqrt{1-\theta_1^2} \Gamma_1 + \theta_1 \Gamma_2  }{ \sqrt{ (\Gamma_1^2 + \Gamma_2^2 + |\Gamma_3|^2)/n }} \right ) \right ].
\end{equation}
The assumption $\|f\|_2 = 1$ amounts to
\begin{equation}\label{eq:fL2}
\EE \left [ h \left ( \frac{\Gamma_1}{ \sqrt{(|\Gamma_3|^2 + \Gamma_1^2 + \Gamma_2^2) / n }} \right )^2 \right ] = 1.
\end{equation}
$$~$$The next lemma follows from a direct calculation.
\begin{lemma} \label{lem:com}
	Assume that $n$ is large enough. Let $\Gamma_1,\Gamma_2 \sim \mathcal{N}(0,1)$ and $\Gamma_3 \sim \mathcal{N}(0, \mathrm{I}_{n-2})$ be independent. Let $\tilde \gamma$ be the density of the random varianble $\frac{\Gamma_1}{ \sqrt{(|\Gamma_3|^2 + \Gamma_1^2 + \Gamma_2^2) / n }}$ and let $\gamma$ be the standard Gaussian density. Then
	$$
	\frac{1}{2} \leq \frac{\tilde \gamma(s)}{\gamma(s)} \leq 2, ~~ \forall s \in [-n^{0.1}, n^{0,1}].
	$$
\end{lemma}
Equation \eqref{eq:fL2} and Lemma \ref{lem:com} imply that $\|h\|_{L_2(\gamma)} \leq 2$ and
\begin{equation}\label{eq:hnorm}
\EE \left [ h \left ( \frac{\Gamma_1  }{ \sqrt{|\Gamma_3|^2 / n }} \right )^2 \right ] \leq 2.
\end{equation}
In what follows, we denote by $C_k$ a constant depending only on $k$ whose value may change between different appearances. Since $H_\ell'(x)  = \ell H_{\ell -1}(x)$, for every $\ell$ there exists a constant $C_\ell$ such that any Hermite polynomial $H_\ell$ with $\ell \leq k$ satisfies
$$
|H_\ell (x(1-s)) - H_\ell(x)| \leq s |x| \ell \max_{|y| \leq |x|} |H_{\ell-1} (y)| \leq C_k s (2 + |x|)^k,~~  \forall s \in (0,1).
$$
Moreover since $h$ is a polynomial of degree at most $8k$ with $\|h\|_{L_2(\gamma)} \leq 2$, we conclude that
\begin{equation}\label{eq:hlip}
|h(x(1-s)) - h(x)| \leq C_k s (2 + |x|)^{8k},~~  \forall s \in (0,1).
\end{equation}
So we can write
$$
b_k(f; \tilde{\theta})[e_{i_1}',...,e_{i_k}'] = \EE \left [  H_{\alpha} (\Gamma_1) J_I(\Gamma_3) h \left ( \frac{\Gamma_1  }{ \sqrt{|\Gamma_3|^2 / n }} \right ) \right ] + T_{res} [e_{i_1}',...,e_{i_k}']
$$
where, relying on \eqref{eq:pres1} and on \eqref{eq:term1decomposed},
$$
T_{res} = \EE \left [H^{(k)} (\Gamma_2 \tilde{\theta} + \Gamma_1 e_1' + \Gamma_3) \left ( h \left ( \frac{\Gamma_1  }{ \sqrt{|\Gamma_3|^2 / n }} \right ) - h \left ( \frac{ \sqrt{1-\theta_1^2} \Gamma_1  }{ \sqrt{(\Gamma_1^2 + |\Gamma_3|^2) / n }} \right ) \right ) \right ]
$$
By Parseval's inequality, we have
\begin{align*}
\|T_{res}\|_2^2 & = \EE \left [ \left ( h \left ( \frac{ \Gamma_1  }{ \sqrt{|\Gamma_3|^2 / n }} \right ) - h \left ( \frac{ \sqrt{1-\theta_1^2} \Gamma_1  }{ \sqrt{(\Gamma_1^2 + |\Gamma_3|^2)/n }} \right ) \right )^2 \right ] \\
& \stackrel{ \eqref{eq:hlip} }{\leq } C_k \EE \left [ \left ( \left |\frac{ \frac{ \sqrt{1-\theta_1^2} \Gamma_1  }{ \sqrt{\Gamma_1^2 + |\Gamma_3|^2 }} - \frac{\Gamma_1  }{ \sqrt{|\Gamma_3|^2}}  }{ \frac{\Gamma_1  }{ \sqrt{|\Gamma_3|^2}}} \right | (2+|\Gamma_1|)^{8k} \right )^2 \right ]  \\
& = C_k \EE \left [ \left ( \left | \frac{ \sqrt{1-\theta_1^2} }{ \sqrt{\frac{\Gamma_1^2}{|\Gamma_3|^2} + 1}  } - 1   \right | (2+|\Gamma_1|)^{8k} \right )^2 \right ] \\
& \leq C_k \EE \left [ \left ( \left (\theta_1^2 + \frac{\Gamma_1^2}{|\Gamma_3|^2} \right ) (2+|\Gamma_1|)^{8k} \right )^2 \right ] \leq C_k \left (\theta_1^4 + \frac{1}{n^2} \right ).
\end{align*}
In a similar manner, \eqref{eq:hlip} and \eqref{eq:term2decomposed} imply that
$$
b_k(f)[e_{i_1}',...,e_{i_k}'] = \EE \left [  H_{\alpha(I)} (\Gamma_1) J_I(\Gamma_3) h \left ( \frac{ \sqrt{1-\theta_1^2} \Gamma_1 + \theta_1 \Gamma_2  }{ \sqrt{|\Gamma_3|^2 / n }} \right ) \right ] + T_{res}'[e_{i_1}',...,e_{i_k}']
$$
with $\|T_{res}'\|_2^2 \leq C_k \left (\theta_1^4 + \frac{1}{n^2} \right )$. Note, however, that since $H_{\alpha(I)}$ is an eigenvector of the heat operator, we have
\begin{align*}
\EE \left [  H_{\alpha(I)} (\Gamma_1) J_I(\Gamma_3) h \left ( \frac{ \sqrt{1-\theta_1^2} \Gamma_1 + \theta_1 \Gamma_2  }{ \sqrt{|\Gamma_3|^2/n}} \right ) \right ] & = \EE \left [ J_I(\Gamma_3) \EE \left . \left [ H_{\alpha(I)} (\Gamma_1)  h \left ( \frac{ \sqrt{1-\theta_1^2} \Gamma_1 + \theta_1 \Gamma_2  }{ \sqrt{|\Gamma_3|^2/n}} \right ) \right | \Gamma_3 \right ] \right ] \\
& = (1-\theta_1^2)^{\alpha(I) / 2} \EE \left [  H_{\alpha(I)} (\Gamma_1) J_I(\Gamma_3) h \left ( \frac{ \Gamma_1}{ \sqrt{|\Gamma_3|^2/n}} \right ) \right ].
\end{align*}
We conclude that
\begin{align*}
b_k(f; \tilde{\theta})[e_{i_1}',...,e_{i_k}'] - b_k(f)[e_{i_1}',...,e_{i_k}'] & = T_{res}[e_{i_1}',...,e_{i_k}'] - T_{res}'[e_{i_1}',...,e_{i_k}'] \\
& + \left ( 1 - (1-\theta_1^2)^{\alpha(I) / 2} \right ) \EE \left [  H_{\alpha(I)} (\Gamma_1) J_I(\Gamma_3) h \left ( \frac{ \Gamma_1}{ \sqrt{|\Gamma_3|^2/n}} \right ) \right ],
\end{align*}
Now, by Parseval,
\begin{align*}
\sum_{I = (i_1,...,i_k) \in [n-1]^k} \left ( 1 - (1-\theta_1^2)^{\alpha(I)} \right )^2 & \EE \left [  H_{\alpha(I)} (\Gamma_1) J_I(\Gamma_3) h \left ( \frac{ \Gamma_1}{ \sqrt{|\Gamma_3|^2/n}} \right ) \right ]^2 \\
& \leq k^2 \theta_1^4 \EE \left [ h \left ( \frac{ \Gamma_1}{ \sqrt{|\Gamma_3|^2/n}} \right )^2 \right ] \stackrel{ \eqref{eq:hnorm} }{\leq} C_k \theta_1^4,
\end{align*}
Combining the last two displays with equation \eqref{eq:basisnorm}, we finally attain
\begin{align*}
\| \mathrm{P}_{\tilde{\theta}^\perp} (b_k(f; \tilde{\theta}) - b_k(f)) \|_{HS}^2 \leq C \theta_1^4 + 4 \|T_{res}'\|_2^2 + 4 \|T_{res}\|_2^2 \leq C_k \left ( \theta_1^4 + \frac{1}{n^2} \right ).
\end{align*}
Since $\EE \theta_1^4 = O(1/n^2)$, taking expectation over $\theta$ establishes \eqref{eq:nts}, and completes the proof of Theorem~\ref{thm:conc}.

\subsection{Loose ends}

\begin{proof} [Proof of Proposition~\ref{p:density}]
	Denote by $\sigma_{n}$ the unique rotationally-invariant measure on the unit sphere in $\RR^n$. A standard calculation (see~\cite[Equation (24)]{EK08}) shows that the density of an $\ell$-dimensional marginal of $\sigma_n$ has the form
	$$
	\psi_{n,\ell}(x) = \psi_{n,\ell}(|x|) = \Gamma_{n,\ell} \left ( 1-|x|^2 \right )^{\frac{n-\ell-2}{2}}, ~~ |x| \leq 1
	$$
	for a constant $\Gamma_{n,\ell}$. So we have by continuity,
	\begin{align*}
	\lim_{\eps \to 0} \frac{1}{\eps} \EE \mathbf{1} \{ |\langle x, \theta \rangle| \leq \eps \} ~& = \lim_{\eps \to 0} \frac{1}{\eps} \EE \mathbf{1} \left \{ |\langle x/|x|, \theta \rangle| \leq \frac{\eps}{|x|} \right  \} = \frac{2}{|x|} \Gamma_{n,1}.
	\end{align*}
	By the continuity of $g$ it follows that
	$$
	\lim_{\eps \to 0} \frac{1}{\eps} \EE \left [ \mathbf{1} \{ |\langle x, \theta \rangle| \leq \eps \} g(\theta) \right ] = \lim_{\eps \to 0} \frac{1}{\eps} \EE \left [ \mathbf{1} \{ |\langle x, \theta \rangle| \leq \eps \} g \left (\frac{\mathrm{Proj}_{x^\perp} \theta}{|\mathrm{Proj}_{x^\perp} \theta|} \right ) \right ],
	$$
	and the first part of the proposition follows by symmetry to revolution about $x$. Now, for the second part, for $\rho \in [0,1]$ denote
	$$
	V(\rho) = \mathrm{Vol} \left ( \left \{ (x,y): ~~ |x| < 1, |\rho x + \sqrt{1-\rho^2} y| < 1 \right \} \right ),
	$$
	the volume of the rhombus with angle $\arcsin(\rho)$ and height $2$. A calculation shows that for all $\rho < 1/2$,
	$$
	V(\rho) = \frac{4}{\sqrt{1-\rho^2}}.
	$$
	So we have by continuity
	\begin{align*}
	\lim_{\eps \to 0} \frac{1}{\eps^2} \EE \left [g(\theta) \mathbf{1} \{ |\langle x, \theta \rangle| \leq \eps, |\langle y, \theta \rangle| \leq \eps \} \right ] ~& = \lim_{\eps \to 0} \frac{1}{|x||y| \eps^2} \EE \left [ g\left  (\frac{\mathrm{Proj}_{x^\perp \cap y^\perp} \theta}{|\mathrm{Proj}_{x^\perp \cap y^\perp} \theta|} \right ) \mathbf{1} \{ |\langle \hat x, \theta \rangle| \leq \eps, |\langle \hat y, \theta \rangle| \leq \eps \} \right ] \\
	& = \frac{\Gamma_{n,2} V(\langle \hat x, \hat y \rangle)}{|x||y|} \EE \left [ g\left  (\frac{\mathrm{Proj}_{x^\perp \cap y^\perp} \theta}{|\mathrm{Proj}_{x^\perp \cap y^\perp} \theta|} \right ) \right ].
	\end{align*}
	The proposition follows.
\end{proof}

\begin{proof} [Proof of Proposition~\ref{p:tensorpoly}]
	Both expressions are invariant to orthogonal transformations applied to both $x,y$, and are therefore functions of $\langle x,y \rangle$, $|x|$ and $|y|$. By applying a rotation, assume that
	\begin{equation}\label{eq:assumpxy}
	x \in \mathrm{span}(e_1), ~~ y \in \mathrm{span}(e_1,e_2), ~~ x_1 \geq 0, ~~ y_2 \geq 0.
	\end{equation}
	Evidently, for any fixed $\theta$ and indices $i_1,...,i_k \in [n]^k$, the expression
	$$
	\mathrm{P}_{\theta^\perp} H^{(k)}(x)[e_{i_1},...,e_{i_k}] \mathrm{P}_{\theta^\perp} H^{(k)}(y) [e_{i_1},...e_{i_k}]
	$$
	is a polynomial of degree at most $k$ in $x_1,y_1,y_2$ with coefficients depending only on $k$. Since the distribution of $\theta_1,\theta_2$ does not depend on $x,y$ given the above assumption, we have that restricted to \eqref{eq:assumpxy}, the two expressions
	$$
	\EE \langle \mathrm{P}_{\theta_{1,2}^\perp} H^{(k)}(x), \mathrm{P}_{\theta_{1,2}^\perp} H^{(k)}(y) \rangle_{HS},
	$$
	are polynomials of degree at most $k$ in $x_1,y_1,y_2$ with coefficients bounded by $O_k(n^k)$. Note that under \eqref{eq:assumpxy}, we have
	$$
	x_1 = |x|, ~~ y_1 = \rho(x,y) |y|, ~~ y_2 = \sqrt{1-\rho(x,y)^2} |y|.
	$$
	Thus, we can express the above expressions as polynomials of degree at most $2k$ in $|x|$, $|y|$, $\rho(x,y)$ and $\sqrt{1-\rho(x,y)^2}$ as long as \eqref{eq:assumpxy} holds. Since the above expressions are invariant under rotations, these forms will hold true in general. This completes the proof.
\end{proof}

\section*{Acknowledgements}

We are thankful to Subhash Khot and Bo'az Klartag for discussions. We are especially grateful to Bo'az for suggesting to use Schur's Lemma which is key to the proof of Theorem~\ref{t:concentration}.

\bibliographystyle{plain}
\bibliography{bibfile}

\appendix

\section{Integrality Gap for Subspaces Near-Intersection}\label{s:sdp-ig}

In this section we sketch an integrality gap instance for the Subspaces Near-Intersection problem. That is, we show that the semidefinite program that minimizes $$\Expc{(u,v)\in E}{| Proj_{\Theta_e^\perp}(\sigma(u))-Proj_{\Theta_e^\perp}(\sigma(v))|_2^2}$$
cannot solve Subspaces Near-Intersection.

Consider the following graph $G = (V,E)$: its vertices correspond to all unit vectors $v\in\R^k$ where coordinates are taken up to sufficiently large precision with respect to $\delta>0$. The subspace associated with the vertex is the one that is spanned by $v$. For the vertex corresponding to vector $v$ there is an edge that touches it for every unit vector $\Theta\in \R^k$ (up to the aforementioned precision) and it connects it to a vertex associated with a random vector $u\in \R^k$ such that $\card{v_{|\Theta^\perp} - u_{|\Theta^\perp}}_2 \approx \sqrt{\delta}$ (the approximation reflects the precision error).
Note that this instance of Subspaces Near-Intersection has a vector solution given by the unit vector associate with every vertex, and it achieves value approximately $\delta$ by construction.
Nevertheless, there is no feasible assignment $\sigma:V\to\R^{k}$ where $|Proj_{\Theta_e^\perp}(\sigma(u))-Proj_{\Theta_e^\perp}(\sigma(v))|_2$ is typically $0.001\sqrt{\delta}$, simply because only the prescribed unit vector is in the subspace of each vertex.

\end{document}